\documentclass[11pt]{article}
\usepackage[svgnames]{xcolor}
\usepackage{tikz}
\usepackage[utf8]{inputenc}
\usepackage[T1]{fontenc}
\usepackage{stmaryrd}
\usepackage{textcomp,amssymb,amsmath,amsthm,amsfonts}
\usepackage{setspace}

\makeatletter
\newtheorem*{rep@theorem}{\rep@title}
\newcommand{\newreptheorem}[2]{%
\newenvironment{rep#1}[1]{%
 \def\rep@title{#2 \ref{##1}}%
 \begin{rep@theorem}}%
 {\end{rep@theorem}}}
\makeatother

\makeatletter
\newtheorem*{rep@proposition}{\rep@title}
\newcommand{\newrepproposition}[2]{%
\newenvironment{rep#1}[1]{%
 \def\rep@title{#2 \ref{##1}}%
 \begin{rep@proposition}}%
 {\end{rep@proposition}}}
\makeatother

\makeatletter
\newtheorem*{rep@lemma}{\rep@title}
\newcommand{\newreplemma}[2]{%
\newenvironment{rep#1}[1]{%
 \def\rep@title{#2 \ref{##1}}%
 \begin{rep@lemma}}%
 {\end{rep@lemma}}}
\makeatother

\def\removevs{}  % set this to 1 to get assignment wo solns  3 gives  solutions

\newcommand{\forceremove}[1]{}

\ifnum 0 \removevs=3 % 
\newcommand{\removing}[1]{\textcolor{blue}{#1}}

\newenvironment{Solns}{{\bf Begin Solution:} \newline \newline}{{\bf End Solution.} \newline}
\newcommand{\honours}[1]{#1}

\fi

\ifnum 0 \removevs=1  %MATH560
\newcommand{\removing}[1]{}
\newenvironment{Solns}{}{}

\newcommand{\honours}[1]{}  %dont print honours question

\fi

\def\appendixvs{1}  % set this to 1 to get assignment wo solns  3 gives  solutions

\ifnum\appendixvs=3 % in this case MATH 560
\newcommand{\appie}[1]{\textcolor{blue}{Move to appendix #1}}

\newenvironment{Solns}{{\bf Begin Solution:} \newline \newline}{{\bf End Solution.} \newline}
\newcommand{\honours}[1]{#1}

\fi

\ifnum\appendixvs=1  %MATH560
\newcommand{\appie}[1]{\textcolor{blue}{moved to appendix}}

\newcommand{\honours}[1]{}  %dont print honours question

\fi

\usepackage{times,mathptmx,courier}
\usepackage[scaled=.92]{helvet}

\usepackage{amssymb}
\usepackage{amsmath}
\usepackage{amsfonts}

\usepackage{bbm}

\usepackage{pifont}

\usepackage{checkend}	% better error messages on left-open environments
\usepackage{graphicx}	% for incorporating external images

\usepackage{booktabs}

\usepackage{listings}
\usepackage[printonlyused,nohyperlinks]{acronym}
\usepackage{color}
\definecolor{greytext}{gray}{0.5}

\usepackage{comment}
\usepackage{ulem}

\usepackage[numbers,sort&compress]{natbib}

\usepackage[compact]{titlesec}
\titleformat*{\section}{\singlespacing\raggedright\bfseries\Large}
\titleformat*{\subsection}{\singlespacing\raggedright\bfseries\large}
\titleformat*{\subsubsection}{\singlespacing\raggedright\bfseries}
\titleformat*{\paragraph}{\singlespacing\raggedright\itshape}

\usepackage[format=hang,indention=-1cm,labelfont={bf},margin=1em]{caption}

\usepackage{url}
\usepackage[colorlinks=false]{hyperref}

\newtheorem{theorem}{Theorem}[section] %
\newreptheorem{theorem}{Theorem}
\newtheorem{defn}[theorem]{Definition} 
\newtheorem{cor}[theorem]{Corollary}
\newtheorem{fact}[theorem]{Fact}
\newtheorem{lemma}[theorem]{Lemma}
\newreplemma{lemma}{Lemma}
\newtheorem{prop}[theorem]{Proposition}
\newrepproposition{proposition}{Proposition}
\newtheorem{conj}[theorem]{Conjecture}
\newtheorem{prob}[theorem]{Problem}

\DeclareUrlCommand\DOI{}
\newcommand{\doi}[1]{\href{http://dx.doi.org/#1}{\DOI{doi:#1}}}


\title{Diversity Embeddings and the Hypergraph Sparsest Cut}
\usepackage{geometry}
\begin{document}

%%%%%%%%%%%%%%%%%%%%%%%%%%%%%%%%%%%%%%%%%%%%%%%%%%
%% From Thesis Components: Tradtional Thesis
%% <http://www.grad.ubc.ca/current-students/dissertation-thesis-preparation/order-components>

% Preliminary Pages (numbered in lower case Roman numerals)
%    1. Title page (mandatory)
%\maketitle

\begin{center}
{\Large\bf Diversity Embeddings and the Hypergraph Sparsest Cut}
\end{center}

\vspace*{0.5cm}

\begin{center}
    Adam D. Jozefiak\footnote{Massachusetts Institute of Technology, \url{jozefiak@mit.edu}}, \hspace{1cm} F. Bruce Shepherd\footnote{University of British Columbia, \url{fbrucesh@cs.ubc.ca}}
\end{center}

\vspace*{0.5cm}

\begin{abstract}
Good approximations have been attained for the sparsest cut problem by rounding solutions to convex relaxations via low-distortion metric embeddings \cite{linial1995geometry, aumann1998log}. Recently, Bryant and Tupper showed that this approach extends to the hypergraph setting by formulating a linear program whose solutions are so-called diversities which are rounded via diversity embeddings into $\ell_1$ \cite{BTGeom}. Diversities are a generalization of metric spaces in which the nonnegative function is defined on all subsets as opposed to only on pairs of elements. 

We show that this approach yields a polytime $O(\log{n})$-approximation when either the supply or demands are given by a graph. This result improves upon Plotkin et al.’s $O(\log{(kn)}\log{n})$-approximation \cite{plotkin1993bounds}, where $k$ is the number of demands, for the setting where the supply is given by a graph and the demands are given by a hypergraph. Additionally, we provide a polytime $O(\min{\{r_G,r_H\}}\log{r_H}\log{n})$-approximation for when the supply and demands are given by hypergraphs whose hyperedges are bounded in cardinality by $r_G$ and $r_H$ respectively. 

To establish these results we provide an $O(\log{n})$-distortion $\ell_1$ embedding for the class of diversities known as diameter diversities. This improves upon Bryant and Tupper’s $O(\log^2{n})$-distortion embedding \cite{BTGeom}. The smallest known distortion with which an arbitrary diversity can be embedded into $\ell_1$ is $O(n)$. We show that for any $\epsilon > 0$ and any $p>0$, there is a family of diversities which cannot be embedded into $\ell_1$ in polynomial time with distortion smaller than $O(n^{1-\epsilon})$ based on querying the diversities on sets of cardinality at most $O(\log{n^p})$, unless $P=NP$. This disproves (an algorithmic refinement of) Bryant and Tupper's conjecture that there exists an $O(\sqrt{n})$-distortion $\ell_1$ embedding based off a diversity's induced metric. 

% In addition, we demonstrate via hypergraph cut sparsifiers that it is sufficient to develop a low-distortion embedding for diversities induced by sparse hypergraphs to obtain good approximations for the sparsest cut in hypergraphs.
\end{abstract}

\vspace*{1cm}
\newpage

\section{Introduction}\label{sec:IntroIntro}

The sparsest cut problem is a fundamental problem in theoretical computer science. In this paper we consider the sparsest cut problem in the most general hypergraph setting. That is, where an instance is defined by two hypergraphs, a supply hypergraph and a demand hypergraph.

Let $G = (V,E_G,w_G)$ be a hypergraph with node set $V$, hyperedge set $E_G$, and nonnegative hyperedge weights $w_G:E_G\to \mathbb{R}_+$ and $H = (V,E_H,w_H)$ be a hypergraph with node set $V$, hyperedge set $E_H$, and nonnegative hyperedge weights $w_H:E_H\to\mathbb{R}_+$. We refer to $G$ and $H$ as the {\em supply hypergraph} and the {\em demand hypergraph}, respectively. In the following, the {\em rank} of a hypergraph is the cardinality of a largest hyperedge. We use the notation $r_G$ for the rank of $G$ and $r_H$ for the rank $H$.
In settings where $G$ or $H$ is a graph ($r_G=2$ or $r_H=2$) we state explicitly that $G$ is a \emph{supply graph} and $H$ is a \emph{demand graph}, respectively.

Similar to graph settings, a cut is determined by $A \subseteq V$ such that $A \neq \emptyset, V$. We say that a  hyperedge $J$ of ($G$ or $H$) is {\em cut} by $A$ if $J \cap A, J \cap A^C \neq \emptyset$; the {\em cut induced by $A$} consists of these edges.  The {\em sparsity} of a cut,  relative to supply $G$ and demand $H$, is then the ratio of the supply edges cut by $A$ and the demand edges cut. 
More formally:
\begin{equation*}
\phi(A) = \frac{\sum_{U\in E_G}w_G(U)\mathbbm{1}_{\{U\cap A \neq \emptyset, U\}}}{\sum_{S\in E_H}w_H(S)\mathbbm{1}_{\{S \cap A \neq \emptyset, S\}}}.
\end{equation*}
 The  {\em sparsest cut} of $G$ and $H$ is then defined as
\begin{equation}\label{eqn:SCDefnEqn}
\phi = \min_{A \subseteq V:A \neq \emptyset,V} \phi(A) . %= \min_{A \subseteq V:A \neq \emptyset,V} \frac{\sum_{U \in E_G} w_G(U) \mathbbm{1}_{\{U \cap A \neq \emptyset, U\}}}
%{\sum_{S \in E_H} w_H(S) \mathbbm{1}_{\{S \cap A \neq \emptyset, S\}}}
\end{equation}

\noindent
If the underlying hypergraphs are ambiguous we may use the notation $\phi_{G,H}$ and $\phi_{G,H}(A)$.

Computing the sparsest cut is NP-hard, even when $G$ and $H$ are graphs \cite{matula1990sparsest}. Consequently, there is a rich history of approximation algorithms for the sparsest cut problem in the graph setting which has culminated in an $O(\sqrt{\log{n}}\log{\log{n}})$-approximation factor for general supply and demand graphs \cite{ALN}. A common approach to achieve such bounds begins with formulating a convex relaxation, such as a linear program (LP) or a semidefinite program (SDP), and then rounding its optimal solution to obtain an integral, but approximate, sparsest cut. Often, solutions of these convex relaxations form a metric space for the nodes $V$ and the rounding step involves embedding this metric into another metric space such as the $\ell_1$ or the $\ell_2$ metric. 

Bryant and Tupper \cite{Bryant_2012} have described a framework which extends such ``metric-relaxations'' to the hypergraph setting. This results in a convex relaxation whose solutions are vectors $\delta$ which assign nonnegative values to arbitrary subsets as opposed to only pairwise distances $d(u,v)$ from a metric space $(V,d)$. The vector $\delta$ does satisfy certain triangle inequalities and the ordered tuples $(V,\delta)$ are termed ``diversities'' by Bryant and Tupper
(see Definition~\ref{defn:Diversity}). Analogous to the approach based on metric embeddings into $\ell_1$ \cite{linial1995geometry, aumann1998log}, one can extract an approximate sparsest cut in the hypergraph setting via diversity embeddings into so-called $\ell_1$ diversities. Our work's focus is largely concerned with low-distortion embeddings of diversities into $\ell_1$ and their application to the sparsest cut in hypergraphs.

In the remainder of this section we list our contributions.
 Formal definitions for some of the objects are only  introduced later.

\subsection{Approximating the Sparsest Cut in Hypergraphs} 

Our contributions for the sparsest cut problem are in the setting where the supply and demand hypergraphs are  general hypergraphs, namely Theorem \ref{thm:GeneralAlgSparsestCut_orig}.
 This more general hypergraph setting has received relatively little investigation except for recent work in the setting where $G$ is a supply hypergraph and $H$ is a demand graph  \cite{LMSSE&SSVE,louis2015HgraphMarkov,kapralov2021towards}.

\begin{theorem}\label{thm:GeneralAlgSparsestCut_orig}(Section \ref{sec:AlgorithmicImplications})
Let $G=(V,E_G,w_G)$ be a supply hypergraph with rank $r_G$ and $H = (V,E_H,w_H)$ be a demand hypergraph with rank $r_H$. Then there is a randomized polynomial-time $O(\min\{r_G,r_H\}\log{n}\log{r_H})$-approximation algorithm for the sparsest cut of $G$ and $H$.
\end{theorem}

We follow the basic approach of Bryant and Tupper   \cite{BTGeom} who gave a similar but existential bound. In order to establish a polytime algorithm we must lose an $O(\log{r_H})$ factor in order to solve (separate) the LP relaxation in polytime (see Section~\ref{sec:micostHSP}).  This is the first  polynomial-time approximation algorithm in the general setting. 

We next consider the restricted case where the supply is a graph, but the demands arise from a hypergraph.

\begin{theorem}\label{ApproxAlgGaGraph}(Section \ref{sec:AlgorithmicImplications})
Let $G=(V,E_G,w_G)$ be a supply graph (equivalently rank $r_G = 2$) and $H = (V,E_H,w_H)$ be a demand hypergraph. Then there is a randomized polynomial-time $O(\log{n})$-approximation algorithm for the sparsest cut of $G$ and $H$.
\end{theorem}

   Theorem \ref{ApproxAlgGaGraph} is an improvement upon an existing polynomial time $O(\log{n}\log{(|E_H|r_H}))$-approximation algorithm due to Plotkin et al. \cite{plotkin1993bounds}. Specifically, the approximation factor guaranteed by Plotkin et al.'s algorithm has a logarithmic dependence on $|E_H|$, which is exponentially large in general, while Theorem \ref{ApproxAlgGaGraph} truly  guarantees a polylogarithmic approximation factor. We mention that  hypergraph demands arise in VLSI optimization research  \cite{grotschel1996packing}.

\subsection{Diversity Embeddings}

We shall see that Theorem \ref{thm:GeneralAlgSparsestCut_orig} relies critically on the following embedding result. It  improves upon the $O(\log^2{n})$-distortion embedding  of a diameter diversity (Definition~\ref{defn:DiameterDiversity}) into the $\ell_1$ diversity, $(\mathbb{R}^{O(\log{n})},\delta_1)$ (see Definition~\ref{defn:ell1Diversity}) in \cite{BTGeom}.

\begin{theorem}\label{thm:DiamDivEmbeddingNew_orig}(Section \ref{sec:DiamDivEmbed})
Let $(X,\delta_{\text{diam}})$ be a diameter diversity with induced metric $(X,d)$ and with $|X| = n$. Then there exists a randomized polynomial time embedding of $(X,\delta_{\text{diam}})$ into the $\ell_1$ diversity $(\mathbb{R}^{O(\log^2{n})},\delta_1)$ with distortion $O(\log{n})$.
\end{theorem}

In contrast, we  give bad news in terms of embedding a general diversity.

\begin{theorem}\label{thm:InapproxResult_orig}(Section \ref{sec:Inapprox})
For any $p \geq 0$ and for any $\epsilon > 0$, there does not exist a polynomial-time diversity $\ell_1$ embedding that has oracle access to a diversity and queries it only on sets of cardinality at most $O(\log{n^p})$ and achieves a distortion of $O(n^{1-\epsilon})$, unless P=NP.
\end{theorem}

 Theorem \ref{thm:InapproxResult_orig} provides an inapproximability result for embedding diversities into $\ell_1$ based only on their induced metric space. This answers an algorithmic-refinement of Bryant and Tupper's conjecture that there exists an $O(\sqrt{n})$-distortion embedding into $\ell_1$ solely using the induced metric of a diversity \cite{BTopen}, see Conjecture \ref{conj:algBTconjecture}.

\subsection{The Minimum Cost Hypergraph Steiner Problem}
\label{sec:micostHSP}

Solving the diversity relaxation for hypergraph sparsest cuts requires a subroutine for separating over exponentially many ``hypergraph Steiner tree constraints''.  The following result is used to establish an $O(\log{r_H})$-approximate separation oracle.

\begin{theorem}\label{thm:HSP_orig}(Section \ref{ch:HSP})
There exists a polynomial time $O(\log{n})$-approximation algorithm for the minimum cost hypergraph Steiner problem (HSP). Specifically, for a hypergraph $G = (V,E,w)$ with nonnegative hyperedge weights $w:E\to\mathbb{R}_+$ and for a set of Steiner nodes $S\subseteq V$, $HSP(G,S)$ can be approximated up to a factor of $O(\log{|S|})$ in polynomial time.
\end{theorem}

To our knowledge, this is the first polynomial-time approximation algorithm for  the {\em minimum-cost hypergraph Steiner} problem. That is, the problem of finding the minimum hyperedge-weight connected subhypergraph that contains a set of Steiner nodes $S$. This is a generalization of the classical Steiner tree problem. One special case which has been studied   is the spanning connected sub-hypergraph problem for which a (tight) logarithmic approximation is known \cite{baudis2000approximating}.

\subsection{A Diversity Embedding Conjecture}
\label{sec:divembed}

While there exist polylogarithmic distortion embeddings for special types of diversities (such as diameter diversities) we are missing a result for general diversities. Indeed Theorem~\ref{thm:InapproxResult_orig} is an inapproximability result in that  case. We  propose a conjecture which offers a work-around. 

It turns out that any diversity can be implicitly defined by some (weighted) hypergraph as we describe now.   Given a hypergraph $H = (V, E)$ and $S \subseteq V$ we let $\mathcal{T}_{(H,S)}$ denote the collection of subsets of hyperedges that induce  connected subhypergraphs whose nodes include $S$.

\begin{defn}[Hypergraph Steiner Diversity]\label{defn:HSteinerDiversity}
Let $H = (V, E, w)$ be a hypergraph with node set $V$, hyperedge set $E$, and nonnegative hyperedge weights $w:E\to \mathbb{R}_+$. The hypergraph Steiner diversity $(V,\delta_\mathcal{H})$ is defined as

\begin{equation*}
\delta_\mathcal{H}(A) = \begin{cases}
\min_{t\in\mathcal{T}_A}\sum_{U\in t}w(U) & \text{ if $|A| \geq 2$} \\
0 & \text{otherwise}.
\end{cases}
\end{equation*} 
\end{defn}

Specifically, any diversity $(X,\delta)$ is a hypergraph Steiner diversity for a hypergraph $H=(X,E,w)$ where $E = \{S \subseteq X: |S| \geq 2\}$ and $w(S) = \delta(S), \forall S\in E$.

We conjecture that there is a diversity embedding into $\ell_1$ with distortion polylogarithmic in the encoding size of any hypergraph which induces the embedding.  For our applications, the hypergraph weights are {\em polynomially bounded}, i.e., have encoding size polynomial in $n=|V(H)|$ and $m=|E(H)|$.

\begin{conj}\label{conj:OurConj}
There is an  algorithm which embeds an arbitrary diversity $(V,\delta_\mathcal{H})$ into the $\ell_1$ diversity with distortion polylogarithmic in the  encoding size of $\mathcal{H}$. If the weights are polynomially bounded, this distortion is $\operatorname{poly}(\log(m,n))$.\footnote{At the time of press, Jeff Kline informed us of an explicit family of hypergraphs that appears to refute Conjecture~\ref{conj:OurConj} as stated; see~\cite{kline2026covering}.}

\end{conj}

 This conjecture is partly motivated by the fact that hypergraph cut sparsifiers allow us to replace an instance of the sparsest cut on dense hypergraphs by an instance on sparse hypergraphs. This can be achieved in polynomial time and with an arbitrarily small approximation loss using \cite{khannacutsparsifier}. Hence the diversity that arises as the solution  to the sparsest-cut convex-relaxation is one that is defined by a sparse hypergraph. See \cite{Jozefiak_2022} for a complete application of hypergraph cut sparsifiers to the sparsest cut problem. We believe that Conjecture~\ref{conj:OurConj} is a more tractable target than a polylogarithmic distortion that is independent of $|E|$. 
 
 % We believe that this sparse structure may make it easier to obtain low-distortion embeddings. 

\iffalse
\textcolor{blue}{I dont understand the sparsity angle it seems. We know that in general we need the distortion to depend on log(m) for embedding diversities.  I think you are saying that for solving sparsest cuts however we may get away with only needing to look at diversities which are defined by sparse hypergraphs? I think it could be a bit cripser/clearer}
\fi

\section{Related Work and History}\label{sec:RelatedWork}

The most significant work for general (hypergraph) sparsest cut is for the case when $G$ is a hypergraph and $H$ is a demand graph (in contrast to Theorem~\ref{ApproxAlgGaGraph} where the supply and demand roles are reversed).
In this setting, Louis \cite{louis2015HgraphMarkov} shows a randomized polytime  $O(\sqrt{\log{n}\log{r_G}}\log{\log{n}})$ approximation, where $r_G$ is the rank of the supply hypergraph $G$.

{\bf Expansion.} A closely related property to the sparsity of a cut is the {\em expansion} of the cut. This is defined as the ratio of the weight of the (hyper)edges crossing the cut to the number of nodes in the smaller partition of the cut. Formally,

\begin{defn}\label{defn:ExpansionOfCut}
Let $G=(V,E,w)$ be a (hyper)graph with nonnegative hyperedge weights $w:E\to \mathbb{R}_+$ and let $A\subseteq V$ such that $A\neq \emptyset,V$. Then the expansion of the cut $A$, $\Phi(A)$, is defined as
\begin{equation*}
\Phi(A) = \frac{\sum_{U\in E}w(U)\mathbbm{1}_{\{U\cap A\neq \emptyset,U\}}}{\min\{|A|,|A^C|\}}.
\end{equation*}
\end{defn} Then the (hyper)graph expansion of $G$ is defined as

\begin{defn}\label{defn:ExpansionProb}
Let $G=(V,E)$ be (hyper)graph with nonnegative hyperedge weights $w:E\to \mathbb{R}_+$. Then the expansion of the (hyper)graph $G$, denoted by $\Phi$, is defined as
\begin{equation*}
\Phi = \min_{A\subseteq V:A\neq\emptyset ,V}\Phi(A) = \min_{A\subseteq V:A\neq\emptyset,V}\frac{\sum_{U\in E}w(U)\mathbbm{1}_{\{U\cap A\neq \emptyset,U\}}}{\min\{|A|,|A^C|\}}.
\end{equation*}
\end{defn}

Often, when discussing (hyper)graph expansion, the supply (hyper)graph $G$ has unit weights and the expansion of a cut is simply the ratio of the number of (hyper)edges crossing the cut to the number of nodes in the smaller partition of the cut. The expansion of a cut is closely related to the sparsity of the cut when the demand network is a complete graph with unit weights,  called \emph{uniform demands}. The problem of computing the minimum expansion cut, even with uniform demands, is NP-hard \cite{matula1990sparsest}.
It is easy to verify that for uniform demands, and $A\subseteq V$ we have that
\begin{equation*}
\sum_{S\in E_H}w_H(S)\mathbbm{1}_{\{S\cap A\neq \emptyset,S\}} = |A||A^C|.
\end{equation*}

\noindent
Thus for each $A\subseteq V$ with $|A| \leq |A^C|$, we have that $\frac{n}{2}|A| \leq |A||A^C| \leq n |A|$. Hence up to a factor of $2$, approximability of the sparsest cut and the minimum expansion cut is equivalent. In particular,  existence of an $O(\alpha)$-approximation algorithm for the sparsest cut is also an $O(\alpha)$-approximation algorithm for the (hyper)graph expansion problem and vice-versa.

{\bf Sparsest cut when $G$ and $H$ are both graphs.}  Leighton and Rao gave an $O(\log n)$ approximation algorithm for the sparsest cut problem with uniform demands \cite{leighton1999multicommodity}. Moreover, they show that the flow-cut gap in this setting is $\Omega(\log{n})$ due to an example where the supply graph is a unit-capacity constant-degree expander graph. This established that  the flow-cut gap for the uniform demand multicommodity flow problem and the sparsest cut problem  is $\Theta(\log{n})$. 
Later, Linial, London, and Rabinovich \cite{linial1995geometry} and Aumann and Rabani \cite{aumann1998log} generalize Leighton and Rao's result to the setting of  arbitrary demand graphs. In particular, they show that the flow-cut gap is bounded by the minimum distortion for embedding  finite metric spaces into an $\ell_1$ metric. According to Bourgain, the minimum distortion of embedding a finite metric into $\ell_1$ is $O(\log{n})$. However, before \cite{linial1995geometry} it remained  open  as to whether this distortion factor is tight; this is answered in the affirmative due to  Linial et al.'s work  coupled with a Leighton and Rao lower bound. Moreover, Linial et al.  also gave the first polynomial time implementation of Bourgain's $O(\log{n})$-distortion $\ell_1$ embedding.

One can strengthen the metric relaxation of the sparsest cut problem,  by adding  negative-type metric constraints. This leads to an SDP relaxation. A metric space $(X,d)$ has negative type if $(X,\sqrt{d})$ embeds isometrically into $\ell_2$. We use $\ell_2^2$ to denote the class of such metrics. The minimum distortion of embedding these metrics into  $\ell_1$  yields the integrality gap for this SDP approach, and a polynomial time  embedding algorithm   yields a polynomial time approximation algorithm for the sparsest cut problem. This is the approach of Arora, Rao, and Vazirani \cite{ARV} who give an $O(\sqrt{\log{n}})$-approximation algorithm for the sparsest cut problem with uniform demands. However, their approach is a randomized rounding scheme of their SDP relaxation as opposed to directly embedding an arbitrary $\ell_2^2$ metric into $\ell_1$.
 Chawla, Gupta, and R\"{a}cke prove that a finite $\ell_2^2$ metric can indeed be polytime embedded into the $\ell_1$ metric with distortion $O(\log^{\frac{3}{4}}{n})$  \cite{chawla2005embeddings}.   Arora, Lee, and Naor \cite{ALN} further improve this by showing that a finite $\ell_2^2$ metric can be embedded into $\ell_1$ with distortion $O(\sqrt{\log{n}}\log{\log{n}})$; this yields a corresponding approximation for  sparsest cut.
 %into  they gave a polynomial time approximation algorithm  to the generalized??? no hypergraphs???  %sparsest cut problem that achieves this approximation factor \cite{ALN}. This %metric embedding result is tight up to a factor of $O(\log{\log{n}})$. 

{\bf Sparsest cut when $G$ is a hypergraph and $H$ is a graph.}
In this setting, Kapralov et al.  \cite{kapralov2021towards} give a polynomial time $O(\log n)$ approximation algorithm for the hypergraph expansion problem and the sparsest cut problem with uniform demands.  They use a rounding scheme of an LP relaxation of the hypergraph expansion problem that optimizes over pseudo metrics. We note that Kapralov et al.'s result actually comes after Louis and Makarychev's $O(\sqrt{\log{n}})$ approximation  \cite{LMSSE&SSVE}.
Kapralov et al.'s approximation algorithm is an LP (as opposed to an SDP) and  is motivated by the construction of spectral sparsifiers for hypergraphs.
%and they require an approximation algorithm for the hypergraph expansion problem %in order to extract cuts with approximate minimal expansion. 
The randomized polytime $O(\sqrt {\log{n}})$-approximation algorithm of 
Louis and Makarychev  \cite{LMSSE&SSVE} for
hypergraph expansion (and hence  the sparsest cut problem with uniform demands)  matches the approximation factor for graph expansion and sparsest cut with uniform demand.
Similar to the progression of \cite{ARV} to \cite{ALN},  Louis gives a randomized polynomial time approximation algorithm for the sparsest cut problem that obtains an $O(\sqrt{\log{r_G}\log{n}}\log{\log{n}})$ approximation factor \cite{louis2015HgraphMarkov}. 
%Louis' technique is inspired by \cite{ALN}, whereby Louis' algorithm solves an %SDP relaxation of the sparsest cut problem, with negative type metric %constraints, where its solution, a metric of the negative type, is embedded into %$\ell_2$ with distortion $O(\sqrt{\log{n}}\log{\log{n}})$ according to %\cite{ALN}'s metric embedding result. Finally, Louis' algorithm performs %randomized rounding to obtain an approximate sparsest cut, incurring an %additional $O(\sqrt{\log{r_G}})$ approximation factor.

{\bf Sparsest cut when $G$ is a graph and $H$ is a hypergraph.}  Plotkin et al. provided a polynomial time $O(\log{n}\log{(|E_H|r_H)})$-approximation algorithm. Their  approach rounds a fractional solution of an LP relaxation without the use of metric embeddings. To our knowledge, there is no polynomial-time algorithm in the hypergraph demand setting that utilizes metric embeddings.  We are not aware of any other polynomial-time approach in this setting, or indeed the  general setting where $G$ and $H$ are arbitrary hypergraphs.

{\bf Sparsest cut when $G$ and $H$ are both hypergraphs.} Bryant and Tupper \cite{Bryant_2012} introduce diversities as a generalization of metric spaces where instead of a nonnegative function defined on pairs of elements, it is defined on arbitrary finite sets of elements. They have  developed a substantial theory on diversities \cite{Bryant_2012, BTGeom, BTopen, bryant2016constant} including around the notion of  diversity embedding. For instance, they  demonstrate that several types of diversities attain polynomial-time low-distortion embeddings into $\ell_1$. Most pertinent to our discussion, in \cite{BTGeom} they generalized the work of Linial et al. \cite{linial1995geometry} whereby  the flow-cut gap in the hypergraph setting is equivalent to the minimum distortion of embedding some finite diversity into the $\ell_1$ diversity. Notably, this work provides an approach to approximating the sparsest cut in hypergraphs However, Bryant and Tupper did not focus on the tractability of this framework and in particular, whether   their approach   has a polynomial-time implementation for  the hypergraph setting.

\subsection{Organization}\label{sec:Organization}

% In Section \ref{ch:diversities} we introduce the notion of a diversity and their properties relevant  to this work.
% In Section \ref{ch:embeddings} we introduce   diversity embeddings;    this section includes proofs of Theorem \ref{thm:DiamDivEmbeddingNew_orig}, in Section \ref{sec:DiamDivEmbed}, and Theorem \ref{thm:InapproxResult_orig}, in Section \ref{sec:Inapprox}.

% In Section \ref{ch:ApproxAlg} we describe the approach of approximating the hypergraph sparsest cut via 
% an LP relaxation whose solutions are diversities. One then  extracts an approximate sparsest cut via embedding such diversities into $\ell_1$. In Section \ref{sec:characterizingSolutions} we characterize the optimal solutions of the LP relaxation. This is used  in the proofs of Theorems \ref{thm:GeneralAlgSparsestCut_orig} and \ref{ApproxAlgGaGraph} (see Section \ref{sec:AlgorithmicImplications}).
% In Section \ref{ch:HSP} we establish tractability of the LP by showing an algorithm for the minimum-cost hypergraph Steiner problem (Theorem \ref{thm:HSP_orig}). 
% %\textcolor{red}{I think we delete this. We already include the %conjecture above in Section 1.4. Is there more to say?  Or perhaps we %want to move Section 1.4 to the end of the paper? Finally, we conclude %with Conjecture \ref{conj:OurConj}. }

Section~\ref{ch:diversities} introduces diversities and the properties
and embedding tools used throughout the paper. Our algorithms for
sparsest cut combine an LP relaxation over diversities with rounding
through embeddings into $\ell_1$.

Section~\ref{ch:ApproxAlg} develops this framework.
Sections~\ref{sec:lp}--\ref{sec:algoverview} formulate the relaxation,
describe the rounding procedure, and summarize their guarantees in
Theorem~\ref{thm:SparsestCutAlGGeneral}.
Section~\ref{sec:characterizingSolutions} characterizes the relevant
optimal diversities, identifying which embedding results can be applied.
Section~\ref{sec:AlgorithmicImplications} combines these
characterizations with the embedding and optimization results
established later to prove
Theorems~\ref{thm:GeneralAlgSparsestCut_orig}
and~\ref{ApproxAlgGaGraph}.

Section~\ref{ch:embeddings} develops the embedding results.
Section~\ref{sec:Inapprox} introduces the hardness construction
underlying Theorem~\ref{thm:InapproxResult_orig}, whose full proof
appears in Appendix~\ref{sec:omitted}.
Section~\ref{sec:DiamDivEmbed} proves the diameter diversity embedding
of Theorem~\ref{thm:DiamDivEmbeddingNew_orig} and extends it to
$k$-diameter diversities, supplying embedding guarantees used in
Section~\ref{ch:ApproxAlg}.

Section~\ref{ch:HSP} proves Theorem~\ref{thm:HSP_orig},
an approximation algorithm for the minimum cost hypergraph Steiner
problem. Appendix~\ref{sec:applications} uses this result to
approximately evaluate hypergraph Steiner diversities and approximately
solve the LP relaxation, providing the computational ingredients
required by the framework. Appendix~\ref{sec:omitted} collects the
remaining omitted proofs.

\section{Diversities}
\label{ch:diversities}

Recently introduced by Bryant and Tupper \cite{Bryant_2012}, diversities are a generalization of metric spaces where instead of a nonnegative function defined on pairs of elements, it is defined on arbitrary finite sets of elements. 

\begin{defn}\label{defn:Diversity}
A diversity is a pair $(X, \delta)$ where $X$ is a set and $\delta$ is a real-valued function defined over the finite subsets of $X$ satisfying the following three axioms:

\begin{enumerate}
    \item $\forall A\in \mathcal{P}(X), \delta(A) \geq 0$
    \item $\delta(A) = 0 $ if and only if $|A| \leq 1$
    \item $\forall A,B,C\in \mathcal{P}(X), C\neq \emptyset \Rightarrow \delta(A\cup B) \leq \delta(A\cup C) + \delta(B\cup C)$.
\end{enumerate}
\end{defn}

We refer to the third axiom  as the \emph{triangle inequality}. It is the property that makes diversities a generalization of  metrics.
Similar to the notion of a pseudo-metric there is a definition of a pseudo diversity. Formally, $(V,\delta)$ is a \emph{pseudo diversity} if $(X,\delta)$ satisfies the three axioms of Definition \ref{defn:Diversity} with the second axiom being weakened to

\begin{equation*}
\delta(A) = 0 \text{ if $|A| \leq 1$}.
\end{equation*} 

\noindent
Henceforth when we refer to $(X, \delta)$ as being a ``diversity'' we mean that $(X, \delta)$ is a pseudo diversity; we make this assumption since  the  LP relaxation optimizes over pseudo diversities.

We note that any (pseudo) diversity $(X,\delta)$ yields an \emph{induced} (pseudo) metric space $(X,d)$ where for any $x,y\in X$ we define $d$ as \begin{equation*}
    d(x,y) = \delta(\{x,y\}).
\end{equation*} Conversely, a metric can induce an infinite number of diversities for whom it is the induced metric. 
 
\begin{defn}[$\mathcal{D}_{(X,d)}$]\label{defn:MetricFamily}
For a metric space $(X,d)$, $\mathcal{D}_{(X,d)}$ is the family of diversities for whom $(X,d)$ is their induced metric. For contexts where the set $X$ is unambiguous the notation $\mathcal{D}_d$ may be used.
\end{defn}

A consequence of the triangle inequality is that (pseudo) diversities are monotone increasing set functions.

\begin{prop}\label{prop:Monotonicity}[see Appendix \ref{sec:omitted}]
(Pseudo) diversities are monotone increasing. That is, if $(X, \delta)$ is a (pseudo) diversity, then for any $A,B\in \mathcal{P}(X)$ we have that

\begin{equation*}
    A\subseteq B \Rightarrow \delta(A) \leq \delta(B).
\end{equation*}
\end{prop}

\subsection{Examples of Diversities}

As described in Section~\ref{sec:divembed}, any diversity can be defined 
by some (sometimes many) hypergraph.   In this section, we define several important  examples of diversities with a considerably simpler structure.

\subsubsection{The $\ell_1$ and Cut Diversities}

Just as the hypergraph Steiner diversity is a generalization of the shortest path metric of a graph, the $\ell_1$ diversity is a natural generalization of the $\ell_1$ metric. For context, we first define the $\ell_1$ metric.
Recall that a metric space $(\mathbb{R}^m, d)$  is the $\ell_1$ metric if for any $x,y \in \mathbb{R}^m$ $d(x,y) = \sum_{i=1}^m |x_i - y_i|$.

\iffalse
\begin{defn}
Let $(\mathbb{R}^m, d)$ be a metric space. Then $(\mathbb{R}^m, d)$ is the $\ell_1$ metric if for any $x,y \in \mathbb{R}^m$

\begin{equation*}
    d(x,y) = \sum_{i=1}^m |x_i - y_i|
\end{equation*}
\end{defn}
\fi

Similarly, the $\ell_1$ diversity is defined as follows 

\begin{defn}[$\ell_1$ Diversity]\label{defn:ell1Diversity}
$(\mathbb{R}^m,\delta_1)$ is an \emph{$\ell_1$ diversity} if for any  $A\in\mathcal{P}(\mathbb{R}^m)$
\begin{equation*}
\delta_1(A) = \sum_{i=1}^m\max_{a,b\in A}|a_i-b_i|.
\end{equation*}
\end{defn}

Likewise, the cut pseudo metric is generalized by the cut pseudo diversity. 

\begin{defn}[Cut Diversity]\label{defn:CutDiversity}
Let $X$ be a set and let $U\subseteq X$ be a nonempty subset. Then $(X,\delta_U)$ is a \emph{cut pseudo diversity} (induced by $U$) if for any $A\in \mathcal{P}(X)$
\begin{equation*}
\delta_U(A) = \begin{cases}
1 & \text{ if $A\cap U\neq \emptyset, A$ } \\
0 &  \text{ otherwise}.
\end{cases}
\end{equation*}
\end{defn}

A classic property of the $\ell_1$ metric is that it can be represented as sum of cut pseudo metrics. Similarly, there is an analogous generalization, due to Bryant and Tupper \cite{BTGeom}, for the $\ell_1$ diversity and cut pseudo diversities, see Theorem \ref{thm:Ell1CutDiv}.

\subsubsection{Extremal Diversities}

 Given a (pseudo) metric space, $(X, d)$, a natural extension of the (pseudo) metric space to a (pseudo) diversity is to define a diversity $(X, \delta)$ where $\delta(A)$ is simply the diameter of the set $A$ in the metric space $(X,d)$. Unsurprisingly,  these  are referred to as diameter diversities.  

\begin{defn}[Diameter Diversity]\label{defn:DiameterDiversity}
Given a (pseudo) metric space $(X,d)$, a (pseudo) \emph{diameter diversity} $(X,\delta_{\text{diam}})$ satisfies for each $A\in\mathcal{P}(X)$

\begin{equation*}
    \delta_{\text{diam}}(A) = \max_{x,y\in A}d(x,y).
\end{equation*}
\end{defn}

Alternatively, one case also extend the (pseudo) metric $(X, d)$ to a (pseudo) diversity based on the minimum cost Steiner trees in the complete graph induced by the weights $d(u,v)$. These are a special case of the hypergraph Steiner diversity in which the hyperedge-weighted hypergraph $H = (V, E, w)$ is simply a graph.

\begin{defn}[Steiner Diversity]\label{defn:SteinerDiversity}
Let $G=(V,E,w)$ be a graph with node set $V$, edge set $E$, and nonnegative edge weights $w:E\to \mathbb{R}_+$. Then the \emph{Steiner diversity} $(V,\delta_{\text{Steiner}})$ is a Steiner diversity if for any $A\in\mathcal{P}(V)$ we have 

\begin{equation*}
\delta_{\text{Steiner}}(A) = \begin{cases}
\min_{t\in\mathcal{T}_A}\sum_{e\in t}w(e) & \text{ if $|A| \geq 2$} \\
0 & \text{otherwise}.
\end{cases}
\end{equation*} where $\mathcal{T}_A$ is the collection of subsets of edges that correspond to connected subgraphs of $G$ that contain the nodes $A$ (cf. Section~\ref{sec:divembed}). Or in other words, $\delta_{\text{Steiner}}(A)$ is defined to be the minimum weight of a subtree of $G$ containing the nodes $A$. 
\end{defn}

A special case of Steiner diversities and a generalization of tree metrics are the tree diversities; Bryant and Tupper introduced these and refer to them as   \emph{phylogenetic diversities}.

\begin{defn}[Tree Diversity]\label{defn:TreeDiversity}
Let $G=(V,E,w)$ be a graph with node set $V$, edge set $E$, and nonnegative edge weights $w:E\to \mathbb{R}_+$. Let $(V,\delta_{\text{Steiner}})$ be the corresponding Steiner diversity. If $(V,E)$ is a tree then $(V,\delta_{\text{Steiner}})$ is a \emph{tree diversity}, for which we may use the notation $(V, \delta_{\text{tree}})$.
\end{defn}

The following extremal result characterizes the (pseudo) diameter diversity as being the ``minimal" (pseudo) diversity and characterizes the (pseudo) Steiner diversity as being the ``maximal" (pseudo) diversity among the family of (pseudo) diversities sharing an induced (pseudo) metric.

\begin{theorem}[Bryant and Tupper in \cite{BTGeom}]
\label{thm:ExtremalResults}
Let $(X,d)$ be a (pseudo) metric space. Let $\mathcal{D}_{(X,d)}$ be the family of diversities for whom $(X,d)$ is their induced (pseudo) metric space. Let $(X,\delta_{\text{diam}})\in\mathcal{D}_{(X,d)}$ be the (pseudo) diameter diversity and $(X,\delta_{\text{Steiner}}) \in \mathcal{D}_{(X,d)}$.  Then for any $(X,\delta)\in \mathcal{D}_{(X,d)}$ and any $A\in\mathcal{P}(X)$ it follows that

\begin{equation*}
\delta_{\text{diam}}(A) \leq \delta(A) \leq \delta_{\text{Steiner}}(A).
\end{equation*} %Or in other words $(X,\delta_{\text{diam}})$ is the minimal (pseudo) diversity of the family $\mathcal{D}_{(X,d)}$ and $(X,\delta_\text{Steiner})$ is the maximal (pseudo) diversity of the family $\mathcal{D}_{(X,d)}$. 
\end{theorem}

\forceremove{
\begin{proof}
Let $(X,\delta)\in \mathcal{D}_{(X,d)}$ and $A\in\mathcal{P}(X)$ be arbitrary. Let $x,y \in A$ be such that 

\begin{equation*}
    d(x,y) = \max_{u,v\in A}d(u,v)
\end{equation*} Let $t\subseteq X\times X$ such that 

\begin{equation*}
    \sum_{(u,v)\in t} d(u,v) = \min_{t'\in \mathcal{T}_A} \sum_{(u,v)\in t'} d(u,v)
\end{equation*}

Then it follows that

\begin{align*}
    \delta_{\text{diam}}(A) & = \max_{u,v\in A}d(u,v) & \text{by definition of $\delta_{\text{diam}}$} \\
    & = d(x,y) & \text{by choice of $x$ and $y$} \\
    & = \delta(\{x,y\}) & \text{by $(X,\delta)\in\mathcal{D}_{(X,d)}$} \\
    & \leq \delta(A) & \text{by $\delta$ being increasing, Proposition \ref{prop:Monotonicity}} \\
    & \leq \sum_{(u,v)\in t}d(u,v) & \text{by the triangle inequality, Proposition \ref{prop:Monotonicity}} \\
    & = \min_{t'\in \mathcal{T}_A} \sum_{(u,v)\in t'} d(u,v) & \text{by choice of $t$} \\
    & = \delta_{\text{Steiner}}(A) & \text{by definition of $\delta_{\text{Steiner}}$}
\end{align*}
This completes the proof.
\end{proof}
}

We also use the notion of $k$-diameter diversity which generalizes the diameter diversity and satisfies a similar extremal result to Theorem \ref{thm:ExtremalResults}. Bryant and Tupper introduced this definition and used  the name \emph{truncated diversity}. 

\begin{defn}[$k$-Diameter Diversity]\label{defn:kdiamdiv}
Given a (pseudo) diversity $(X,\delta)$ and $k\in \mathbb{Z}_{\geq 0}$, we say that $(X,\delta_{\text{$k$-diam}})$ is the \emph{$k$-diameter (pseudo) diversity} of $(X,\delta)$ if for any $A\in\mathcal{P}(X)$

\begin{equation}\label{eqn:kdiamdiv}
\delta_{\text{$k$-diam}}(A) = \max_{B\subseteq A: |B| \leq k} \delta(B).
\end{equation} Furthermore, if a (pseudo) diversity satisfies equation (\ref{eqn:kdiamdiv}), then it is referred to as a $k$-diameter (pseudo) diversity.
\end{defn}

Additionally, we  generalize  $\mathcal{D}_{(X,d)}$ to a $k$-diameter setting. 

\begin{defn}[$\mathcal{D}_{(X,\delta,k)}$]\label{defn:kDiameterFamily}
For a (pseudo) diversity $(X,\delta)$ and $k\in \mathbb{Z}_{\geq 0}$, $\mathcal{D}_{(X,\delta,k)}$ is the family of (pseudo) diversities which are equivalent to $\delta$ on subsets of $X$ of cardinality at most $k$. For contexts where the set $X$ is unambiguous the notation $\mathcal{D}_{(\delta,k)}$ may be used.
\end{defn}

 It is easy to see that a diameter diversity is a 2-diameter diversity.

\begin{fact}\label{fact:2DiamDiv}
A 2-diameter (pseudo) diversity is a diameter (pseudo) diversity. Furthermore if $(X,\delta)$ is a diversity with the induced metric space $(X,d)$ then $\mathcal{D}_{(X,d)} = \mathcal{D}_{(X,\delta,2)}$.
\end{fact}

Next, we prove that a $k$-diameter diversity $(X,\delta)$ is the minimal diversity of the family $\mathcal{D}_{(X,\delta,k)}$, generalizing Theorem \ref{thm:ExtremalResults} due to Bryant and Tupper. We use this extremal property of the $k$-diameter diversity in order to characterize the optimal solutions of the diversity-relaxation for the sparsest cut, Theorem \ref{thm:FlowCutGapDiversity}.

\begin{theorem}\label{thm:MinimalityOfkDiameterDiversity}[see Appendix \ref{sec:omitted}]
Given a (pseudo) diversity $(X,\delta)$ and $k\in \mathbb{Z}_{\geq 0}$ we let $(X,\delta_{\text{$k$-diam}})$ be the $k$-diameter (pseudo) diversity of $(X,\delta)$. Then for any $(X,\delta')\in \mathcal{D}_{(X,\delta,k)}$ and any $A\in\mathcal{P}(X)$ it follows that

\begin{equation*}
    \delta_{\text{$k$-diam}}(A) \leq \delta'(A).
\end{equation*} Or in other words $(X,\delta_{\text{$k$-diam}})$ is the minimal (pseudo) diversity of the family $\mathcal{D}_{(X,\delta,k)}$.
\end{theorem}

\subsection{Diversity Embeddings}
\label{ch:embeddings}

In this section we introduce the notion of  diversity embeddings since low distortion embeddings into the $\ell_1$ diversity play a key role in finding an approximation for sparsest cuts. 

\begin{defn}[Diversity Embedding]\label{defn:EmbeddingOfDiv}
Let $(X,\delta_X)$ and $(Y,\delta_Y)$ be two diversities. Let $f$ be a map from $X$ to $Y$. We say that $f$ is an embedding of the diversity $(X,\delta_X)$ into the diversity $(Y,\delta_Y)$ with \emph{distortion} $c$ if there are constants $c_1,c_2 \geq 1$ such that $c = c_1c_2$ and for any $A\in \mathcal{P}(X)$ we have that

\begin{equation}\label{eqn:EmbeddingofDiv}
\frac{1}{c_1}\delta_X(A) \leq \delta_Y(f(A)) \leq c_2\delta_X(A).
\end{equation} For convenience, when referring to such an embedding we may use the notation $(X,\delta_X)\to (Y,\delta_Y)$ and $(X,\delta_X)\xrightarrow{c}(Y,\delta_Y)$, where the latter specifies the distortion of the embedding. We may also precede such notation with ``$f:$'' in order to specify the map. 
\end{defn}

If we restrict (\ref{eqn:EmbeddingofDiv}) in Definition \ref{defn:EmbeddingOfDiv} to sets $A$ of cardinality two, then we recover the exact definition of a metric embedding. This is significant since facts about metric embeddings can  extend to facts about diversity embeddings since each diversity ``encodes'' an induced metric space.

 We use the term \emph{$\ell_1$-embeddable} to mean embedding into $\ell_1$ isometrically, i.e.,  distortion $c=1$. The following analogue for the metric case is also key for our results.

\begin{theorem}[Bryant and Tupper \cite{BTGeom}] \label{thm:Ell1CutDiv}
Let $(X,\delta)$ be a diversity where $|X| = n < \infty$. Then the following are equivalent.
\begin{enumerate}
\item $(X,\delta)$ is embeddable into $\ell_1^m$.
\item $(X,\delta)$ is a nonnegative combination of $O(nm)$ cut diversities.
\end{enumerate}
\end{theorem}

As for the time-complexity of a diversity embedding, we will assume that we have access to an oracle $\alpha_{(X,\delta_X)}:\mathcal{P}(X)\to\mathbb{R}_+$ which computes $\delta_X(A)$ for some subset of $\mathcal{A}\subseteq \mathcal{P}(X)$ in polynomial time with respect to the size of the representation of $(X,\delta_X)$ that we are given. For instance, $(X,\delta_X)$ may be a hypergraph Steiner diversity given by a hyperedge-weighted hypergraph or as a $|\mathcal{P}(X)|$-long enumeration of the values of $\delta_X$. We say that an embedding $f:(X,\delta_X)\to (Y,\delta_Y)$ is a polynomial time embedding if we have an algorithm that computes $f$ in time that is  polynomial in the representation of $(X,\delta_X)$ while querying the oracle $\alpha_{(X,\delta_X)}$ a polynomial number of times on the subsets $\mathcal{A}$. As we demonstrate in Section \ref{sec:Inapprox}, one can construct a family of diversities and a corresponding representation such that any low-distortion embedding is necessarily polytime-intractable.

\subsubsection{A Steiner Diversity Embedding}\label{sec:SteinerDivEmbed}

\iffalse
In this section we present a proof of Bryant and Tupper's $O(\log{n})$-distortion embedding of a Steiner diversity into the $\ell_1$ diversity. In Subsection \ref{sec:HSteinerDivEmbed} we present a corollary of this result which is an $O(k\log{n})$-distortion embedding of a hypergraph Steiner diversity corresponding to a rank-$k$ hypergraph into the $\ell_1$ diversity. 
\fi

Bryant and Tupper give a Steiner diversity embedding utilizing
two ingredients. First is Fakcharoenphol, Rao, and Talwar's probabilistic embedding of $n$-node metric spaces into dominating tree metrics with distortion $O(\log{n})$ in expectation \cite{FRT}, which we refer to as the \emph{FRT algorithm}. Given that tree metrics are isometrically embeddable into the $\ell_1$ metric, the FRT algorithm provides another $O(\log{n})$-distortion metric embedding into $\ell_1$, complementing Bourgain’s earlier construction \cite{bourgain1985lipschitz}. For the theorems below, we assume that $(X,\delta_{\text{tree}})$ and $(X,\delta_{\text{Steiner}})$ are represented by their induced metrics, thus constituting a representation that is $O(n^2)$ in size.

\iffalse

\begin{theorem}[The FRT Algorithm]\label{thm:FRT}
Let $(X,d)$ be a metric space where $|X|=n$. Then there is a randomized polynomial-time algorithm that embeds $(X,d)$ into a dominating tree metric with distortion $O(\log{n})$ in expectation. Specifically, this randomized polynomial-time algorithm produces a tree $T = (X,E,w)$ with nonnegative edge weights $w$ such that the corresponding shortest path metric, which happens to be a tree metric, $(X,d_T)$ satisfies for every $x,y\in X$

\begin{align*}
    & \text{ } d(x,y) \leq d_T(x,y) \label{eqn:FRT1} \\
    & \text{ } \mathbb{E}[d_T(x,y)] \leq O(\log{n})d(x,y) \label{eqn:FRT2}
\end{align*}
\end{theorem}

In addition to the FRT algorithm, Bryant and Tupper's embedding utilizes the fact that tree diversities are isometrically embeddable into the $\ell_1$ diversity \cite{Bryant_2012}. This embedding follows by the fact that, like a tree metric, a tree diversity is a sum of cut pseudo-diversities which correspond to the edges of the tree that defines the diversity.
\fi

The second ingredient is an argument that tree diversities embed isometrically into $\ell_1$.
\begin{theorem}\label{thm:TreeDiv}
Let $(X,\delta_\text{tree})$ be a tree diversity. Then there exists an embedding of $(X,\delta_\text{tree})$ into the $\ell_1$ diversity with distortion $1$. Moreover, this embedding is computable in polynomial-time with respect to $|X|$.
\end{theorem}

%Theorems \ref{thm:FRT} and \ref{thm:TreeDiv}, we present a proof of %Bryant and Tupper's Steiner diversity embedding from \cite{BTopen}.

Theorem~\ref{thm:TreeDiv} and the FRT algorithm 
are then combined to prove the following.

\begin{theorem}[Restatement of Theorem 2 in \cite{BTopen}]\label{thm:SteinerDivEmbedding}
Let $(X,\delta_\text{Steiner})$ be a Steiner diversity where $|X| = n$. Then there is a randomized polynomial-time algorithm that embeds $(X,\delta_\text{Steiner})$ into the $\ell_1$ diversity with $O(\log{n})$ distortion.
\end{theorem}

\subsubsection{A Hypergraph Steiner Diversity Embedding}\label{sec:HSteinerDivEmbed}

 Bryant and Tupper \cite{BTopen} show  that a hypergraph Steiner diversity corresponding to a rank $k$ hypergraph can be embedded into the $\ell_1$ diversity with distortion $O(k\log{n})$. 
At a high level, this $O(k\log{n})$-distortion embedding follows by approximating a hyperedge-weighted rank-$k$ hypergraph with an edge-weighted graph while incurring an $O(k)$ approximation factor. Or in other words, embedding a hypergraph Steiner diversity into a Steiner diversity with distortion $O(k)$. Then, the additional $O(\log{n})$ factor is incurred by embedding this Steiner diversity into the $\ell_1$ diversity. 

\begin{theorem}\label{thm:HSteinerToSteiner}
Let $H = (V,E,w)$ be a rank $k$ hypergraph with node set $V$, edge set $E$, and nonnegative hyperedge weights $w:E\to \mathbb{R}_+$. Let $(V,\delta_\mathcal{H})$ be the corresponding hypergraph Steiner diversity. Then there is a Steiner diversity $(V,\delta_\text{Steiner})$ into which $(V,\delta_{\mathcal{H}})$ can be polynomial-time embedded with distortion $O(k)$.
\end{theorem}

The main result of this section follows as a corollary of Theorems \ref{thm:SteinerDivEmbedding} and \ref{thm:HSteinerToSteiner}.

\begin{cor}[Bryant and Tupper in \cite{BTopen}]\label{cor:HSteinerDivEmbed}
Let $H = (V,E,w)$ be a rank-$k$ hypergraph with node set $V$, edge set $E$, and nonnegative hyperedge weights $w:E\to \mathbb{R}_+$. Let $(V,\delta_\mathcal{H})$ be the corresponding hypergraph Steiner diversity. Then $(V,\delta_\mathcal{H})$ can be embedded into the $\ell_1$ diversity with distortion $O(k\log{n})$, in randomized polynomial-time. 
\end{cor}
\begin{proof}
We let $f:(V,\delta_\mathcal{H})\to (V,\delta_\text{Steiner})$ be a polynomial-time $O(k)$-distortion embedding of $(V,\delta_\mathcal{H})$ into some Steiner diversity $(V,\delta_\text{Steiner})$, due to Theorem \ref{thm:HSteinerToSteiner}. We let $g:(V,\delta_\text{Steiner})\to (\mathbb{R}^m,\delta_1)$ be a randomized polynomial-time $O(\log{n})$-distortion embedding of $(V,\delta_\text{Steiner})$ into the $\ell_1$ diversity for some dimension $m$, due to Theorem \ref{thm:SteinerDivEmbedding}. Then the map $g \circ f$ is an embedding of $(V,\delta_\mathcal{H})$ into the $\ell_1$ diversity with distortion $O(k\log{n})$. Moreover, this embedding is computable in randomized polynomial-time. This completes the proof.
\end{proof}

\subsection{Subadditive Set Functions as Diversities}\label{sec:SubadditiveSetFunc}

In this section we answer the question of how to modify a nonnegative monotone subadditive set function into a diversity? This is a necessary technical lemma for the inapproximability result Theorem \ref{thm:InapproxResult_orig}.

Recall that a set function
$f:2^X\to \mathbb{R}$ is {\em subadditive} if for any $A,B\subseteq X$, $f$ satisfies:

\begin{equation*}
f(A\cup B) \leq f(A) + f(B).    
\end{equation*}

\noindent
Subadditive functions are closely related to diversities if we consider the third axiom of diversities, the triangle inequality. Specifically, that axiom could be recast as requiring subadditivity for sets $A,B$ with non-empty intersection.  (In that case $A \cap B \neq \emptyset$ plays the role of $C$.)  Replacing the third axiom with this condition does not quite give an equivalent definition of diversity however. There are minor nuances, but  nonetheless, there is a close relationship between diversities and subadditive set functions. 

Due to Proposition \ref{prop:Monotonicity},  in order for a subadditive set function to be, or even ``resemble'', a diversity it must be nonnegative and increasing. 
%We need this  in   Section \ref{sec:SubadditiveSetFunc}.

\begin{lemma}\label{lemma:subadditive}[see Appendix \ref{sec:omitted}]
Let $X$ be a set and let $f:2^X\to\mathbb{R}$ be a nonnegative, monotone, and subadditive set function. Then $(X,\delta)$ is a pseudo-diversity where $\delta$ is defined as 

\begin{equation*}
    \delta(A) = \begin{cases}  f(A) & \text{if $|A| \geq 2$} \\ 0 & \text{otherwise.} \end{cases}
\end{equation*}
\end{lemma}

\section{Approximating Sparsest Cuts via Diversity Embeddings}\label{ch:ApproxAlg}

We present the framework of Bryant and Tupper \cite{BTGeom}
for studying sparsest cuts (flow-cut gaps) when  both  the supply and demand are hypergraphs.  This is based on an LP relaxation that optimizes over pseudo-diversities  and then gives a low-distortion embedding  of the diversity corresponding to an optimal LP solution into the $\ell_1$ diversity. An approximate sparsest cut is then extracted due to the fact that $\ell_1$ diversities are a nonnegative sum of cut diversities. Ultimately one can show that the flow-cut gap of the sparsest cut and the maximum concurrent multicommodity flow
(we do not formally define here - see \cite{BTGeom})
 is bounded above by the  distortion of the $\ell_1$ embedding.
We first give an overview of the ingredients.

In Section~\ref{sec:lp} we introduce the relaxation and point out the missing elements to make it polytime solvable.  Section~\ref{sec:rounding} discusses  rounding a solution (a diversity) into an approximate cut and Section~\ref{sec:algoverview} puts the pieces together to give a generic approximation algorithm for sparsest cut. The quality of the approximation depends on how well one can round the diversity produced by the LP.
In Section \ref{sec:characterizingSolutions} we show that  the optimal solution for the  LP relaxation when $G$ is a graph and $H$ is a hypergraph is actually a diameter diversity. In Theorem \ref{thm:DiamDivEmbeddingNew_orig} we then show that such diversities have an $O(\log n)$ polytime embedding into $\ell_1$. A similar bound is achievable for $G$ hypergraph and $H$ a graph due to the optimal diversity being a Steiner diversity.  We do not focus on this since in this case, bounds for sparsest cut due to Louis are superior. Section \ref{sec:AlgorithmicImplications} then gives  algorithmic details which show that this results in  a randomized polynomial-time $O(\log{n})$-approximation algorithm for the sparsest cut problem in the setting where either the supply or demand hypergraphs is a graph.

Throughout this section we consider an instance of the sparsest cut problem defined by a supply hypergraph $G=(V,E_G,w_G)$ with rank $r_G$ and a demand hypergraph $H=(V,E_H,w_H)$ with rank $r_H$.

\subsection{A Linear Programming Relaxation}
\label{sec:lp}

According to the definition of  sparsest cut $\phi$, the definition of a cut pseudo diversity, and Definition \ref{defn:CutDiversity}, we have that the sparsest cut $\phi$ is equivalent to

\begin{equation}\label{eqn:ApproxAlgEqn1}
    \phi = \min_{A\subseteq V} \frac{\sum_{U\in E_G}w_G(U)\delta_{A}(U)}{\sum_{S\in E_H}w_H(S)\delta_A(S)}.
\end{equation}

This can be relaxed to  optimizing over  all pseudo diversities
(not just cuts):

\begin{equation}\label{eqn:ApproxAlgEqn2}
    \min_{(V,\delta) \text{ is a pseudo-diversity}} \frac{\sum_{U\in E_G}w_G(U)\delta(U)}{\sum_{S\in E_H}w_H(S)\delta(S)}.
\end{equation} We call (\ref{eqn:ApproxAlgEqn2}) the \emph{sparsest cut diversity-relaxation}. Moreover, since (\ref{eqn:ApproxAlgEqn2}) is a relaxation of (\ref{eqn:ApproxAlgEqn1}) we have that \begin{equation}\label{eqn:ApproxAlgEqn3}
    \phi \geq \min_{(V,\delta) \text{ is a pseudo-diversity}} \frac{\sum_{U\in E_G}w_G(U)\delta(U)}{\sum_{S\in E_H}w_H(S)\delta(S)}.
\end{equation}

It can be shown (see \cite{BTGeom}) that (\ref{eqn:ApproxAlgEqn2}) is equivalent to the following optimization problem

\begin{equation}\label{eqn:ApproxAlgEqn4}
\begin{aligned}
\min & \sum_{U\in E_G} w_G(U)\delta(U) \\
\text{s.t.} & \sum_{S\in E_H} w_H(S)\delta(S) \geq 1 \\
& (V,\delta) \text{ is a pseudo-diversity}.
\end{aligned}
\end{equation}

Recall that $\mathcal{T}_{(G,S)}$ is the collection of subsets of hyperedges of $G$ that correspond to connected subhypergraphs that contain the nodes $S$. Then (\ref{eqn:ApproxAlgEqn4}) can be shown to be equivalent to the following LP 

\begin{equation}\label{eqn:LP1}
\begin{aligned}
\min & \sum_{U\in E_G}w_G(U)d_U \\
\text{s.t.} & \sum_{S\in E_H}w_H(S)y_S \geq 1 \\
& \sum_{U\in t} d_U \geq y_S & \forall S\in E_H, t\in \mathcal{T}_{(G,S)} \\
& d_U \geq 0 & \forall U\in E_G \\
& y_S \geq 0 & \forall S\in E_H. \\
\end{aligned}
\end{equation} We can take a feasible solution of (\ref{eqn:LP1}), $\{d_U\}_{U\in E_G}$ and $\{y_S\}_{S\in E_H}$, and define a corresponding feasible solution to (\ref{eqn:ApproxAlgEqn4}) with an equivalent objective value. This solution corresponds to  a hypergraph Steiner diversity (see Definition~\ref{defn:HSteinerDiversity})
\begin{equation}
\label{eqn:ApproxAlgEqn5}
    \mbox{$(V,\delta_\mathcal{G})$ where $\mathcal{G}=(V,E_G,w)$ with
  $w:E_G\to\mathbb{R}_+$ is defined as $w(U) = d_U$.}
\end{equation}

This LP relaxation has polynomially many variables\footnote{If $G$ and $H$ are defined explicitly.} but in general  has exponentially many constraints. Specifically, for $S\in E_H$, $\mathcal{T}_{(G,S)}$ may have  exponentially many subsets.  Hence, the set of constraints $\{\sum_{U\in t}d_U \geq y_S\}_{t\in \mathcal{T}_{(G,S)}}$ is exponentially large in general. We can approximately solve this LP however if there is a polynomial time (approximate) separation oracle for \begin{equation*}
    \min_{t\in \mathcal{T}_{(G,S)}}\sum_{U\in t}d_U \geq y_S
\end{equation*} for each $S \in E_H$. This is the \underline{hypergraph Steiner problem (HSP)} which  is NP-complete \cite{baudis2000approximating}. We give an $O(\log{r_H})$-approximate separation oracle (the first as far as we know) in  Section \ref{ch:HSP}. Hence there is an associated approximation algorithm through the Ellipsoid Algorithm framework \cite{carr2000randomized,jansen2003approximate}. In the remainder, we assume that the solutions we obtain for the LP have an encoding size which is polynomially bounded in $m,n$.

\subsection{Rounding the Linear Programming Relaxation}
\label{sec:rounding}

Given a solution to (\ref{eqn:LP1}) we immediately have a solution to the sparsest cut diversity-relaxation (\ref{eqn:ApproxAlgEqn2}), namely
the diversity $(V,\delta)$ from (\ref{eqn:ApproxAlgEqn5}). An approximate sparsest cut can  be extracted as follows.  Let $f:V\to \mathbb{R}^m$ be an embedding from $(V,\delta)$ to the $\ell_1$ diversity $(\mathbb{R}^m,\delta_1)$ with distortion $c \geq 1$. This yields the following inequalities

\begin{equation}\label{eqn:ApproxAlgEqn6}
\frac{\sum_{U\in E_G}w_G(U)\delta_1(f(U))}{\sum_{S\in E_H}w_H(S)\delta_1(f(S))} \leq c \frac{\sum_{U\in E_G}w_G(U)\delta(U)}{\sum_{S\in E_H}w_H(S)\delta(S)} \leq c\phi.
\end{equation} where the first inequality follows by the fact that $f$ is an embedding with distortion $c$ and the second inequality follows by the fact that $(V,\delta)$ is a feasible solution to (\ref{eqn:ApproxAlgEqn2}) which is a relaxation of $\phi$ according to Equation (\ref{eqn:ApproxAlgEqn3}).

By Theorem \ref{thm:Ell1CutDiv}, $\delta_1$ is a nonnegative sum of $O(nm)$ cut pseudo-diversities. Therefore, there exists some collection
$\mathcal{F}$ of subsets of $2^{f(V)}$,  where $|\mathcal{F}| \in O(nm)$, such that for any $A\subseteq V$ \begin{equation*}
    \delta_1(f(A)) = \sum_{B\in \mathcal{F}}\alpha_B\delta_B(f(A)),
\end{equation*} where $\{\alpha_B\}_{B\in \mathcal{F}}$ are positive scalars and $\{(f(V),\delta_B)\}_{B\in \mathcal{F}}$ is a collection of cut pseudo-diversities. This immediately yields that \begin{equation}\label{eqn:ApproxAlgEqn7}
 \frac{\sum_{U\in E_G}w_G(U)\delta_1(f(U))}{\sum_{S\in E_H}w_H(S)\delta_1(f(S))} = \frac{\sum_{U\in E_G}w_G(U)\big[\sum_{B\in \mathcal{F}}\alpha_B\delta_B(f(U))\big]}{\sum_{S\in E_H}w_H(S)\big[\sum_{B\in \mathcal{F}}\alpha_B\delta_B(f(S))\big]}.
\end{equation} Next, rearranging the order of the summations, we have that (\ref{eqn:ApproxAlgEqn7}) is equivalent to \begin{equation}\label{eqn:ApproxAlgEqn7.5}
    \frac{\sum_{B\in \mathcal{F}}\alpha_B\big[\sum_{U\in E_G}w_G(U)\delta_B(f(U))\big]}{\sum_{B\in \mathcal{F}}\alpha_B\big[\sum_{S\in E_H}w_H(S)\delta_B(f(S))\big]}.
\end{equation}

Finally, by the argument in the proof of Bryant and
Tupper~\cite[Proposition~16]{BTGeom}, there exists a cut
$B_0 \subseteq V$, with $f(B_0) \in \mathcal{F}$ and
positive demand denominator, such that

\begin{equation}\label{eqn:ApproxAlgEqn8}
\frac{\sum_{U\in E_G}w_G(U)\delta_{f(B_0)}(f(U))}{\sum_{S\in E_H}w_H(S)\delta_{f(B_0)}(f(S))} \leq  \frac{\sum_{B\in \mathcal{F}}\alpha_B\big[\sum_{U\in E_G}w_G(U)\delta_B(f(U))\big]}{\sum_{B\in \mathcal{F}}\alpha_B\big[\sum_{S\in E_H}w_H(S)\delta_B(f(S))\big]}.
\end{equation} We note that the left side of the inequality (\ref{eqn:ApproxAlgEqn8}) is simply \begin{equation}\label{eqn:ApproxAlgEqn9}
    \phi(B_0) = \frac{\sum_{U\in E_G}w_G(U)\delta_{B_0}(U)}{\sum_{S\in E_H}w_H(S)\delta_{B_0}(S)} = \frac{\sum_{U\in E_G}w_G(U)\delta_{f(B_0)}(f(U))}{\sum_{S\in E_H}w_H(S)\delta_{f(B_0)}(f(S))} .
\end{equation} where the second equality follows by the fact that cuts are preserved under the map $f:V\to \mathbb{R}^m$. Such a cut can be found in polynomial time given the embedding \(f\). Indeed, the family \(\mathcal F\) in (28) can be constructed by sorting \(f(V)\) in each coordinate and taking the cuts between consecutive distinct coordinate values, with coefficients given by the corresponding gaps. This yields at most \(m(n-1)\) cuts. We evaluate the sparsity of each corresponding cut in \(V\), discard those with zero demand denominator, and select one of minimum sparsity.

Hence putting together (\ref{eqn:ApproxAlgEqn6}), (\ref{eqn:ApproxAlgEqn7}), (\ref{eqn:ApproxAlgEqn7.5}), (\ref{eqn:ApproxAlgEqn8}), (\ref{eqn:ApproxAlgEqn9}), and the fact that $\phi \leq \phi(B_0)$ we have that \begin{equation}\label{eqn:ApproxAlgEqn10}
    \phi \leq \phi(B_0) \leq c\phi.
\end{equation}

% \textcolor{red}{Point out Where do we mention that $B_0$ is obtained in poytime? I think perhaps we need a stronger version of Theorem 3.15?}

\subsection{An Approximation Algorithm for the Sparsest Cut}
\label{sec:algoverview}

Combining the LP relaxation in Section~\ref{sec:lp} with the embedding and rounding argument in Section~\ref{sec:rounding} yields the following theorem. The overall approximation factor is obtained by multiplying the losses from approximately solving the LP relaxation, evaluating the associated hypergraph Steiner diversity, and embedding into \(\ell_1\), represented by \(\alpha_{\mathrm{LP}}\), \(\alpha_{\mathrm{div}}\), and \(\alpha_{\mathrm{emb}}\), respectively. The final cut extraction introduces no additional approximation loss. The polynomial-time guarantee follows from the assumed running times of these stages and the cut-extraction procedure described in Section~\ref{sec:rounding}.

\begin{theorem}\label{thm:SparsestCutAlGGeneral}
Let $G=(V,E_G,w_G)$ be a supply hypergraph with rank $r_G$ and $H = (V,E_H,w_H)$ be a demand hypergraph with rank $r_H$. 
\begin{enumerate}
    \item Let $\alpha_{\text{LP}}$ be the approximation factor for some (randomized) polynomial-time algorithm that (approximately) solves LP Relaxation (\ref{eqn:LP1}). 
    \item Let $\alpha_{\text{div}}$ be the approximation factor for some (randomized) polynomial-time algorithm that computes the hypergraph Steiner diversity corresponding to the optimal solution of LP Relaxation \ref{eqn:LP1}, as defined according to Equation (\ref{eqn:ApproxAlgEqn5}).
    \item  Let $\alpha_{\text{emb}}$ be the distortion of embedding the aforementioned diversity into the $\ell_1$ diversity for some (randomized) polynomial time algorithm.
\end{enumerate} Then the approach outlined in this section forms a (randomized) polynomial-time $O(\alpha_{\text{LP}}\alpha_{\text{div}}\alpha_{\text{emb}})$-approximation algorithm for the sparsest cut problem in hypergraphs.
\end{theorem}

\subsection{Optimal Solutions of the Sparsest Cut Diversity-Relaxation with Special Structure}\label{sec:characterizingSolutions}

In this section we describe existence of  optimal solutions of the sparsest cut diversity-relaxation with a specific structure, Equation (\ref{eqn:ApproxAlgEqn2}). The algorithmic implications are discussed in the next section.
We show that if the supply hypergraph is a graph, then the optimal diversity is a Steiner diversity and if the demand hypergraph is a graph, then the optimal diversity is a diameter diversity. Notably, both of these diversities have $O(\log{n})$-distortion polynomial-time embeddings into $\ell_1$.

Key structural properties of optimal diversities are derived in Theorems~\ref{thm:FlowCutGapDiversity} and \ref{thm:GAGraph}.

% \textcolor{red}{FBS: I just stuck in "apprximation to" in the theorem statement.  I guess the point is that in Section 4.5 we are not working with optimal diversities but near-optimal according to whatever the separation algorithm allows us.  I think its fine that in both cases we get the same structure used in Theorems 4.2 and 4.4 but perhaps you can double-check that is correct. Maybe we should be more explicit in the proof about how arbitrary feasible solutions could be used. Not just optimal solutions.}

\begin{theorem}\label{thm:FlowCutGapDiversity}
Let $G = (V,E_G,w_G)$ be a supply hypergraph and let $H=(V,E_H,w_H)$ be a demand hypergraph with rank $r_H$. Let $(V,\delta)$ be a (pseudo) diversity attaining (some approximation to) the optimal objective value to the sparsest cut diversity-relaxation, Equation (\ref{eqn:ApproxAlgEqn2}). Let $(V,\delta_{\text{$r_H$-diam}})$ be the $r_H$-diameter diversity of $(V,\delta)$. Then $(V,\delta_{\text{$r_H$-diam}})$ is also optimal. 
\end{theorem}

\begin{proof}
According to Theorem \ref{thm:MinimalityOfkDiameterDiversity}, $(V,\delta)$ and $(V,\delta_{\text{$r_H$-diam}})$ are both members of $\mathcal{D}_{(V,\delta,r_H)}$, and moreover, $(V,\delta_{\text{$r_H$-diam}})$ is the minimal diversity of this family. By definition of $\mathcal{D}_{(V,\delta,r_H)}$ and by the fact that $\forall S\in E_H, |S| \leq r_H$, it follows that

\begin{equation}\label{eqn:CharacterEqn1}
\sum_{S\in E_H}w_H(S)\delta(S) = \sum_{S\in E_H}w_H(S)\delta_{\text{$r_H$-diam}}(S).
\end{equation}

Now edges $U$ of $G$ may be larger than $r_H$, but in this case we use the minimality of $(V,\delta_{\text{$r_H$-diam}})$ among the family $\mathcal{D}_{(V,\delta,r_H)}$. We have:

\begin{equation}\label{eqn:CharacterEqn2}
\sum_{U\in E_G}w_G(U)\delta_{\text{$r_H$-diam}}(U) \leq \sum_{U\in E_G}w_G(U)\delta(U).
\end{equation}

\noindent
Then (\ref{eqn:CharacterEqn1}) and (\ref{eqn:CharacterEqn2}) imply that 

\begin{equation*}
\frac{\sum_{U\in E_G}w_G(U)\delta_{\text{$r_H$-diam}}(U)}{\sum_{S\in E_H}w_H(S)\delta_{\text{$r_H$-diam}}(S)}\leq \frac{\sum_{U\in E_G}w_G(U)\delta(U)}{\sum_{S\in E_H}w_H(S)\delta(S)}.
\end{equation*}

Therefore, the optimal diversity for the sparsest cut diversity-relaxation, $(V,\delta)$, can be assumed to be its  $r_H$-diameter diversity $(V,\delta_{\text{$r_H$-diam}})$, thus completing the proof.
\end{proof}

\begin{cor}\label{cor:HAGraph}
Let $G = (V,E_G,w_G)$ be a supply hypergraph and let $H=(V,E_H,w_H)$ be a demand graph, that is it has rank $r_H=2$. Let $(V,\delta)$ be a (pseudo) diversity attaining (an approximation to) the optimal objective value to the sparsest cut diversity-relaxation, Equation (\ref{eqn:ApproxAlgEqn2}). Let $(V,d)$ be the induced metric space of $(V,\delta)$ and let $(V,\delta_{\text{diam}}) \in \mathcal{D}_{(V,d)}$ be the diameter diversity whose induced metric space is $(V,d)$. Then $(V,\delta)$ can be assumed to be $(V,\delta_{\text{diam}})$, a diameter diversity.
\end{cor}

\begin{proof}
This corollary follows by Fact \ref{fact:2DiamDiv} and Theorem \ref{thm:FlowCutGapDiversity}.
\end{proof}

\begin{theorem}\label{thm:GAGraph}
Let $G = (V,E_G,w_G)$ be a supply graph, that is it has rank $r_G=2$, and let $H=(V,E_H,w_H)$ be a demand hypergraph. Let $(V,\delta)$ be a (pseudo) diversity attaining (an approximation to) the optimal objective value to the sparsest cut diversity-relaxation, Equation (\ref{eqn:ApproxAlgEqn2}). Let $(V,d)$ be the induced metric space of $(V,\delta)$ and let $(V,\delta_\text{Steiner})\in\mathcal{D}_{(V,d)}$ be the Steiner diversity whose induced metric space is $(V,d)$. Then $(V,\delta)$ can be assumed to be $(V,\delta_\text{Steiner})$,  a Steiner diversity. 
\end{theorem}

\begin{proof}
Let $U\in E_G$ be arbitrary. Since $r_G = 2$ then $|U| = 2$ and we let $U = \{u,v\}$. Then it follows that

\begin{equation}\label{eqn:CharacterEqn3} 
    \delta(U) = d(u,v) = \delta_\text{Steiner}(U).
\end{equation} where the  equalities follow by the fact that both $\delta$ and $\delta_\text{Steiner}$ have induced metric $d$. From this we have that \begin{equation}\label{eqn:CharacterEqn4} 
\sum_{U\in E_G}w_G(U)\delta(U) = \sum_{U\in E_G}w_G(U)\delta_\text{Steiner}(U).
\end{equation}

By Theorem \ref{thm:ExtremalResults}, $(V,\delta_\text{Steiner})$ is the maximal diversity of the family $\mathcal{D}_{(V,d)}$ (diversities with induced metric $d$) and so it follows that

\begin{equation}\label{eqn:CharacterEqn5} 
\sum_{S\in E_H}w_H(S)\delta(S) \leq \sum_{S\in E_H} w_H(S)\delta_\text{Steiner}(S).
\end{equation}

Then (\ref{eqn:CharacterEqn4}) and (\ref{eqn:CharacterEqn5}) imply that

\begin{equation*}
\frac{\sum_{U\in E_G}w_G(U)\delta_\text{Steiner}(U)}{\sum_{S\in E_H}w_H(S)\delta_\text{Steiner}(S)} \leq \frac{\sum_{U\in E_G}w_G(U)\delta(U)}{\sum_{S\in E_H}w_H(S)\delta(S)}.
\end{equation*}

Therefore, the optimal diversity for the sparsest cut diversity-relaxation $(V,\delta)$ can be assumed to be the Steiner diversity $(V,\delta_\text{Steiner})$. This completes the proof.
\end{proof}

\subsection{Algorithmic Implications}\label{sec:AlgorithmicImplications}

Based on this approach, we present a polynomial-time approximation algorithm for the case where the supply and demand hypergraphs are arbitrary hypergraphs with ranks $r_G$ and $r_H$, respectively.

\begin{reptheorem}{thm:GeneralAlgSparsestCut_orig}
%\label{thm:GeneralAlgSparsestCut}
Let $G=(V,E_G,w_G)$ be a supply hypergraph with rank $r_G$ and $H = (V,E_H,w_H)$ be a demand hypergraph with rank $r_H$. Then there is a randomized polynomial-time $O(\min\{r_G,r_H\}\log{n}\log{r_H})$-approximation algorithm for the sparsest cut of $G$ and $H$.
\end{reptheorem}

\begin{proof}
According to Theorem \ref{thm:SparsestCutAlGGeneral} there is a polynomial-time $O(\alpha_{\text{LP}}\alpha_{\text{div}}\alpha_{\text{emb}})$-approximation algorithm for the sparsest cut of $G$ and $H$, where $\alpha_{\text{LP}}$, $\alpha_{\text{div}}$, and $\alpha_{\text{emb}}$ are as defined in Theorem \ref{thm:SparsestCutAlGGeneral}.

According to Corollary \ref{cor:LPRelaxApproxAlg}, there is an $O(\log{r_H})$-approximation algorithm for the LP Relaxation (\ref{eqn:LP1}), hence $\alpha_{\text{LP}} = O(\log{r_H})$. 

We let $(V,\delta)$ be an optimal solution to the sparsest cut diversity-relaxation (\ref{eqn:ApproxAlgEqn2}). We note that $(V,\delta)$ is a hypergraph Steiner diversity corresponding to a rank $r_G$ hyperedge-weighted hypergraph, namely $(V,E_G, w)$ where $w(U) = d_U$ and $\{d_U\}_{U\in E_G}$ are from an optimal solution to LP Relaxation (\ref{eqn:LP1}). Then according to Corollary \ref{cor:HSteinerDivEmbed} there exists a randomized polynomial-time $O(r_G\log{n})$-distortion embedding of $(V,\delta)$ into the $\ell_1$ diversity.

Alternatively, according to Theorem \ref{thm:FlowCutGapDiversity} $(V,\delta)$ is a $r_H$-diameter diversity. Then according to Corollary \ref{cor:kdiamdivembed} there is a randomized polynomial-time $O(r_H\log{n})$-distortion embedding of $(V,\delta)$ into the $\ell_1$ diversity. Thus, we can choose whether to embed $(V,\delta)$ into $\ell_1$ as a hypergraph Steiner diversity or a $r_H$-diameter diversity based on whether $r_G$ or $r_H$ is smaller. Therefore, $\alpha_{\text{emb}} = O(\min\{r_G,r_H\}\log{n})$. 

Since the two embeddings, Corollary \ref{cor:HSteinerDivEmbed} and \ref{cor:kdiamdivembed} only require the induced metric space of $(V,\delta)$,  we only need to compute $\delta(A)$ for $A\in\mathcal{P}(V)$ where $|A| = 2$. Thus, by Corollary \ref{cor:HSteinerDivApproxAlg} we have that $\alpha_{\text{div}} = O(\log{2}) = O(1)$. 

Hence, we have a randomized polynomial-time $O(\min\{r_G,r_H\}\log{n}\log{r_H})$-approximation algorithm for the sparsest cut problem in $G$ and $H$. This completes the proof
\end{proof}

This is the first randomized polynomial-time approximation algorithm for the setting where $G$ and $H$ are arbitrary hypergraphs. An immediate corollary of this result is an $O(\log{n})$-approximation algorithm for the case where the demand hypergraph is simply a graph

\begin{cor}
Let $G=(V,E_G,w_G)$ be a supply hypergraph and $H = (V,E_H,w_H)$ be a demand graph, that is it has rank $r_H = 2$. Then there is a randomized polynomial-time $O(\log{n})$-approximation algorithm for the sparsest cut of $G$ and $H$.
\end{cor}

\begin{proof}
This corollary follows immediately by Theorem \ref{thm:GeneralAlgSparsestCut_orig}  and the fact that $r_H = 2$. 
\end{proof}

\vspace*{0.2cm}
For the setting where $G$ is a hypergraph and $H$ is a graph, there is a randomized polynomial-time $O(\sqrt{\log{r_G}\log{n}}\log\log{n})$-approximation algorithm due to Louis \cite{louis2015HgraphMarkov} which is based on an SDP relaxation. However, among LP-based approaches our algorithm is the first to attain an $O(\log{n})$-approximation when $H$ is an arbitrary graph. Specifically, Kapralov et al. \cite{kapralov2021towards} attain an $O(\log{n})$-approximation when $H$ has uniform demands.
We note that this corollary can be proven alternatively by our characterization of the optimal diversity to the sparsest cut diversity-relaxation (\ref{eqn:ApproxAlgEqn2}) being a diameter diversity for the case when the demand hypergraph is a graph, Corollary \ref{cor:HAGraph}. 

Our characterization Theorem \ref{thm:GAGraph} also underpins a randomized polynomial-time $O(\log{n})$-approximation algorithm for the case where the supply hypergraph is a graph. 

\begin{reptheorem}{ApproxAlgGaGraph}
Let $G=(V,E_G,w_G)$ be a supply graph, that is it has rank $r_G = 2$, and $H = (V,E_H,w_H)$ be a demand hypergraph. Then there is a randomized polynomial-time $O(\log{n})$-approximation algorithm for the sparsest cut of $G$ and $H$.
\end{reptheorem}

\begin{proof}
According to Theorem \ref{thm:SparsestCutAlGGeneral} there is a polynomial-time $O(\alpha_{\text{LP}}\alpha_{\text{div}}\alpha_{\text{emb}})$-approximation algorithm for the sparsest cut of $G$ and $H$, where $\alpha_{\text{LP}}$ $\alpha_{\text{div}}$, and $\alpha_{\text{emb}}$ are as defined in Theorem \ref{thm:SparsestCutAlGGeneral}.

We recall that for each $S\in E_H$ the LP Relaxation (\ref{eqn:LP1}) may have exponentially many constraints of the form \begin{equation*}
    \{\sum_{U\in t}d_U \geq y_S\}_{t\in \mathcal{T}_{(G,S)}}.
\end{equation*} Approximately separating over these constraints amounts to approximating the minimum cost Steiner tree for the nodes $S$. Therefore, we can approximately separate over these constraints in polynomial time using a polynomial-time $O(1)$-approximation algorithm for the minimum-cost Steiner tree problem \cite{kou1981fast,takahashi1980approximate, wu1986faster, byrka2010improved}. Hence $\alpha_{\text{LP}}=O(1)$.

Let $(V,\delta)$ be the (approximately) optimal diversity of the sparsest cut diversity-relaxation. Then according to Theorem \ref{thm:GAGraph} $(V,\delta)$ is a Steiner diversity, and moreover, we can compute $\delta(A)$ for any $A\subseteq V$ in polynomial time up to a factor of $O(1)$, again, by a polynomial-time $O(1)$-approximation algorithm for the Steiner tree problem. Therefore, $\alpha_{\text{div}} = O(1)$. 

Finally, according to Theorem \ref{thm:SteinerDivEmbedding}, there is a randomized polynomial-time $O(\log{n})$-distortion embedding of $(V,\delta)$, a Steiner diversity, into $\ell_1$. Hence $\alpha_{\text{emb}} = O(\log{n})$. 

Hence, we have a randomized polynomial-time $O(\log{n})$-approximation algorithm for the sparsest cut problem in $G$ and $H$. This completes the proof.

\end{proof}

The previous state-of-the-art algorithm for the setting where $G$ is a graph and $H$ is a hypergraph is a polynomial-time $O(\log{n}\log{(|E_H|r_H)}) $-approximation algorithm due to Plotkin et al. \cite{plotkin1993bounds}. Our $O(\log{n})$-approximation is a notable improvement due to the fact that $|E_H|$ may be exponentially large.

In the next section we establish the remaining diversity embedding theorems needed to obtain the preceding results.

\section{Low Distortion Embeddings of Diversities into $\ell_1$ }
\label{ch:embeddings}

The framework from section \ref{ch:ApproxAlg} establishes that approximation for sparsest cut is directly linked to the distortion of embedding a diversity into the $\ell_1$ diversity.  Here we discuss a general investigation originally posed  by Bryant and Tupper \cite{BTopen}.

\begin{prob}\label{prob:BTorig}
Let $(X,\delta)$ be an arbitrary diversity where $|X| = n$. What is the minimum distortion with which $(X,\delta)$ can be embedded into an $\ell_1$ diversity?
\end{prob}

We strengthen their question about existential bounds by additionally asking for a tractable bound, which implicitly requires  the dimension $m$ of the $\ell_1$ diversity $(\mathbb{R}^m,\delta_1)$ to be bounded by a polynomial factor. We formalize this as follows. 

\begin{prob}
Let $(X,\delta)$ be an arbitrary diversity where $|X| = n$. What is the minimum distortion with which $(X,\delta)$ can be embedded into an $\ell_1$ diversity in polynomial-time. 
\end{prob}

A  polynomial-time $O(n)$ distortion embedding is achievable, provided that the induced metric of the diversity can be queried in polynomial-time. 

\begin{theorem}[Theorem 1 in \cite{BTopen}]\label{thm:nDistorEmbed}
Let $(X,\delta)$ be an arbitrary diversity where $|X| = n$. Then there is an embedding of $(X,\delta)$ into the $\ell_1$ diversity, $(\mathbb{R}^n,\delta_1)$, with distortion $n$.
\end{theorem}

As noted in \cite{BTopen}, there is an $\Omega(\log{n})$ lower bound for distortion of $\ell_1$ embeddings. However, they ask whether there exist lower distortion (specifically $O(\sqrt{n})$) embeddings based on a diversity's induced metric. In Section~\ref{sec:Inapprox} we show this is not possible, at least not via a tractable embedding. 
We close the section with a positive result. We give an improved  diameter diversity embedding result. Its extension to $k$-diameter diversities is  needed in  the proof of Theorem~\ref{thm:GeneralAlgSparsestCut_orig}.

\subsection{An Inapproximability Result for Diversity Embeddings}\label{sec:Inapprox}

In this section we provide an inapproximability result for diversity embeddings into $\ell_1$. In short, our result states that there are diversities which cannot be embedded into $\ell_1$ with distortion smaller than $\Omega(n)$ by a polynomial-time algorithm that only queries the induced metric of the diversity, unless P=NP. A consequence of this result is that the $O(n)$-distortion $\ell_1$ embedding, Theorem \ref{thm:nDistorEmbed}, is asymptotically optimal among algorithms that query only the induced metric of a diversity. We start by stating a conjecture of Bryant and Tupper \cite{BTopen} whose algorithmic refinement is disproved by our inapproximability result.

\begin{conj}[Restatement of Bryant and Tupper's Conjecture in \cite{BTopen}]\label{conj:BTconjecture}
There exists an $O(\sqrt{n})$-distortion diversity embedding into $\ell_1$ that is based solely on the induced metric of a diversity.
\end{conj}

Bryant and Tupper conjecture the existence of a diversity embedding that attains low-distortion, specifically $O(\sqrt{n})$, and that this embedding utilizes only a diversity's induced metric. However, Bryant and Tupper do not insist on any algorithmic requirements, notably time complexity. We provide the following algorithmic refinement of their conjecture. 

\begin{conj}[Algorithmic Refinement of Conjecture \ref{conj:BTconjecture}]\label{conj:algBTconjecture}
There exists a polynomial-time $O(\sqrt{n})$-distortion diversity embedding into $\ell_1$ that is based solely on the induced metric of a diversity.
\end{conj}

We refute this refinement with the following inapproximability result.

\begin{reptheorem}{thm:InapproxResult_orig}
%\label{thm:InapproxResult}
For any $p \geq 0$ and for any $\epsilon > 0$, there does not exist a polynomial-time diversity $\ell_1$ embedding that has oracle access to a diversity and queries it only on sets of cardinality at most $O(\log{n^p})$ and achieves a distortion of $O(n^{1-\epsilon})$, unless P=NP.
\end{reptheorem}

This theorem statement is quite cumbersome and so we provide the following corollary which more clearly disproves Conjecture \ref{conj:algBTconjecture}.

\begin{cor}\label{InapproxResultSimplified}
For any $\epsilon >0$, there does not exist a polynomial-time diversity $\ell_1$ embedding that is based solely on the induced metric of a diversity with a distortion of $O(n^{1-\epsilon})$. 
\end{cor}

\iffalse
\begin{proof}
Setting $p = 0$ in Theorem \ref{thm:InapproxResult_orig}, we have that there does not exist a polynomial-time diversity $\ell_1$ embedding that queries a diversity on sets of cardinality at most $O(1)$ and attains a distortion of $O(n^{1-\epsilon})$, unless P=NP. Specifically, sets of cardinality $O(1)$ include sets of cardinality two, or in other words, the induced metric of the diversity.
\end{proof}
\fi

Therefore, Conjecture \ref{conj:algBTconjecture} is disproved. An interesting observation is that existing diversity embeddings into $\ell_1$ are both computable in polynomial-time and are computed solely using the induced metric of a diversity. Notably, these include Bryant and Tupper's $O(n)$-distortion embedding of an arbitrary diversity into $\ell_1$, Theorem \ref{thm:nDistorEmbed}, and the two $O(\log{n})$-distortion embeddings of the diameter and Steiner diversities into $\ell_1$, Theorems \ref{thm:SteinerDivEmbedding} and \ref{thm:DiamDivEmbeddingNew_orig}. This naturally posits the observation that if one were to improve upon the $O(n)$-distortion of embedding an arbitrary diversity into $\ell_1$, say for a hypergraph Steiner diversity, one must construct an algorithm that utilizes the value of the diversity on sets of arbitrary size. 

The proof of Theorem \ref{thm:InapproxResult_orig} rests upon a reduction from the notorious independent set problem to the problem of embedding a diversity into $\ell_1$. In Section \ref{Sec:IndependentSet} we introduce the independent set problem, state an inapproximability result for it, and define the independent set diversity. In Section \ref{sec:omitted} we give the proof of Theorem \ref{thm:InapproxResult_orig}. 

\subsubsection{The Independent Set Diversity}\label{Sec:IndependentSet}

In this section we introduce the independent set problem, state an inapproximability result for it, and conclude with defining the independent set diversity.

\begin{defn}[Independent Set]\label{defn:IndependentSetDiversity}
Given a graph $G=(V,E)$, a subset of nodes $S\subseteq V$ is an independent set of $G$ if for all $u,v\in S$ there does not exist an edge $(u,v)\in E$.
\end{defn}

Thus, the independent set problem is simply the problem of computing the maximum cardinality of an independent set of a graph.

\begin{defn}[Independent Set Problem]
Given a graph $G=(V,E)$, the independent set problem, ISP$(G)$, is defined as
\begin{equation*}
   \text{ISP}(G) = \max \{|S|: \text{ $S$ is an independent set of $G$}\}.
\end{equation*}
\end{defn}

For any $\epsilon >0$, the independent set problem is inapproximable up to a factor of $O(n^{1-\epsilon})$, unless $P=NP$ \cite{arora1998proof, hastad1996clique}.

\begin{theorem}[Inapproximability of the Independent Set Problem]\label{thm:InapproxIS}
For any $\epsilon > 0 $, there does not exist a polynomial-time approximation algorithm for the independent set problem with an approximation factor smaller than $O(n^{1-\epsilon})$, unless P=NP.
\end{theorem}

In order to construct the independent set diversity we require the following technical lemma that characterizes an independent set function as being nonnegative, increasing, and subadditive.

\begin{lemma}\label{lemma:ISLemma}[see Appendix \ref{sec:omitted}]
Let $G = (V,E)$ be a graph. We define the independent set function $f_{IS}:2^V\to \mathbb{Z}_{\geq 0}$ as

\begin{equation*}
    f_{IS}(A) = \max\{|S|: \text{$S\subseteq A$, $S$ is an independent set of $G$}\}.
\end{equation*}

Then the set function $f_{IS}$ is nonnegative, increasing, and subadditive.
\end{lemma}

Next, we define the independent set diversity. By Lemmas \ref{lemma:subadditive} and  \ref{lemma:ISLemma} it follows that this  is in fact a diversity.

\begin{defn}[Independent Set Diversity]
Let $G= (V,E)$ be a graph and let $f_{IS}$ be defined as  in Lemma \ref{lemma:ISLemma}. Then we define the independent set diversity, $(V,\delta_{IS})$, as

\begin{equation*}
    \delta_{IS}(A) = \begin{cases}
    f_{IS}(A) & \text{ if $|A|\geq 2$} \\
    0 & \text{ otherwise}.
    \end{cases}
\end{equation*}
\end{defn}

\begin{proof}
We prove that $(V,\delta_{IS})$ is in fact a diversity. According to Lemma \ref{lemma:ISLemma}, $f_{IS}$ is nonnegative, increasing, and subadditive set function over the ground set $V$. Then, according to Lemma \ref{lemma:subadditive} and the definition of $(V,\delta_{IS})$ it follows that $(V,\delta)$ is a (pseudo)-diversity. In fact, $(V,\delta_{IS})$ is a diversity since for any $A\subseteq V$ where $|A| \geq 2$ it follows that $\delta_{IS}(A) = f_{IS}(A) \geq 1 > 0$.
\end{proof}

Theorem \ref{thm:InapproxResult_orig} now follows by a reduction of the independent set problem to the problem of computing a low-distortion diversity embedding. We point the reader to  Appendix \ref{sec:omitted} for the complete proof.

\subsection{An $O(\log n)$ Diameter Diversity Embedding}\label{sec:DiamDivEmbed}

    In this section we prove  Theorem \ref{thm:DiamDivEmbeddingNew_orig} in Section \ref{sec:DiamEmbeddingNew}. The distortion achieved by the latter embedding is asymptotically optimal given the $\Omega(\log{n})$ lower bound.
    %on the distortion of embedding arbitrary diversities into $\ell_1$, Theorem \ref{thm:DivLowerBound}.

\subsubsection{Fr\'echet Embeddings}

    Our $\ell_1$ embedding of the diameter diversity is based off polynomial-time implementations of Bourgain's original $O(\log{n})$-distortion metric embedding into $\ell_1$ \cite{bourgain1985lipschitz}. 
    
    \begin{theorem}[Restatement of Proposition 1 of  \cite{bourgain1985lipschitz}]\label{thm:Bourgain}
    Let $(X,d)$ be a finite metric space where $|X| = n$. There exists an embedding of $(X,d)$ into the $\ell_1$ metric $(\mathbb{R}^k,d_1)$, where $k\in O(2^n)$, with distortion $O(\log{n})$.
    \end{theorem}
    
    We remark that Bourgain's embedding is a scaled Fr\'echet embedding.

    \begin{defn}
    Let $(X,d)$ be a metric space. A Fr\'echet embedding is a map $f:(X,d)\to (\mathbb{R}^k,d')$ where each coordinate, $f_i:X\to \mathbb{R}$, of the embedding is defined as
    \begin{equation*}
        f_i(x) = d(x,A_i) = \min_{y\in A_i } \text{ } d(x,y) 
    \end{equation*} for some nonempty $A_i\subseteq X$. 
    \end{defn}
    
    Due to the triangle inequality of metrics, a Fr\'echet embedding is coordinate-wise non-expansive, which we define and prove below.
    
    \begin{prop}\label{prop:Frechet}
    Let $(X,d)$ be a metric space and let $f:(X,d)\to(\mathbb{R}^k,d')$ be a Fr\'echet embedding. Then for any coordinate, $i \in \{1,2,\ldots, k\}$ and any $x,y\in X$, it follows that
    
    \begin{equation*}
        |f_i(x) - f_i(y)| \leq d(x,y).
    \end{equation*}
    \end{prop}
    
    \iffalse
    \begin{proof}
    Let $x,y\in X$ and let $i\in\{1, 2, \ldots, k\}$ be arbitrary. Then it follows that
    \begin{align*}
        |f_i(x) - f_i(y)| & = |d(x,A_i) - d(y,A_i)| & \text{for some nonempty $A_i\subseteq X$} \\
        & = |d(x,u) - d(y,u')| & \text{for some $u,u'\in A_i$}
    \end{align*}
    
    Without loss of generality, we assume that $d(x,u) - d(y,u') \geq 0$. Then we have that 
    
    \begin{align*}
        |d(x,u) - d(y,u')| & = d(x,u) - d(y,u') \\
        & \leq d(x,u') - d(y,u') & \text{since $d(x,u) = \min_{v\in A_i}d(x,v)$} \\
        & \leq d(x,y) + d(y,u') - d(y,u') & \text{by the triangle inequality} \\
        & = d(x,y)
    \end{align*}
    This completes the proof.
    \end{proof}
    \fi
    
The original embedding due to Bourgain is an existential result that is algorithmically intractable. Linial, London, and Rabinovich \cite{linial1995geometry}  provided a randomized polynomial-time implementation  by sampling a relatively small subset of the  of coordinate maps 
    
    \begin{equation*}
        \{ f_i(x) = d(x,A_i) \}_{A_i \subseteq X}
    \end{equation*} in order to achieve a randomized polynomial-time complexity.

The fact that Bourgain's embedding is a scaled Fr\'echet embedding is key since  Theorem \ref{thm:DiamDivEmbeddingNew_orig} is based off the following implementation of Bourgain's embedding.

\begin{theorem}[Lemma 3 in \cite{aumann1998log}]\label{thm:LLR2}
Let $(X,d)$ be a metric space with $|X| = n$. Then there exists an embedding, $f:X\to \mathbb{R}^{O(\log^2{n})}$, of $(X,d)$ into the $\ell_1$ metric $(\mathbb{R}^{O(\log^2{n})},d_1)$ with distortion $O(\log{n})$. That is,

\begin{equation}\label{eqn:LLR2Eqn1}
    \frac{1}{O(\log{n})}d(x,y) \leq \|f(x) - f(y)\|_1 \leq d(x,y).
\end{equation} Furthermore, the embedding $f$ is a scaled Fr\'echet embedding where for any  coordinate 
%\newline 
$i \in \{1,2,\ldots, O(\log^2{n})$\}, we have $f_i$ is defined as

\begin{equation}\label{eqn:LLR2Eqn2}
    f_i(x) = \frac{1}{O(\log^2{n})}d(x,A_i)
\end{equation}
where $A_i\subseteq X$.
\end{theorem}

Scaling the Fr\'echet embedding, in the above theorem, is a necessary step. In fact, the original embedding of Bourgain is also scaled proportionally to the exponentially large dimension of the $\ell_1$ metric space being embedded into. 

\subsubsection{The $O(\log{n})$-Distortion Embedding}\label{sec:DiamEmbeddingNew}

\begin{reptheorem}{thm:DiamDivEmbeddingNew_orig}
%\label{thm:DiamDivEmbeddingNew}
Let $(X,\delta_{\text{diam}})$ be a diameter diversity with $|X| = n$
and induced metric $d$. Then there exists a randomized polynomial-time embedding of $(X,\delta_{\text{diam}})$ into the $\ell_1$ diversity $(\mathbb{R}^{O(\log^2{n})},\delta_1)$ with (an optimal) distortion $O(\log{n})$.
\end{reptheorem}

\begin{proof}
By Theorem \ref{thm:LLR2}, we let $f:X\to\mathbb{R}^{O(\log^2{n})}$ be a scaled Fr\'echet embedding of $(X, d)$ into the $\ell_1$ metric $(\mathbb{R}^{O(\log^2{n})}, d_1)$ with distortion $O(\log{n})$. We then define our diversity embedding from $(X,\delta_{\text{diam}})$ to the $\ell_1$ diversity $(\mathbb{R}^{O(\log^2{n})}, \delta_1)$ to be simply the map $f$. Given that Theorem \ref{thm:LLR2} guarantees that $f$ is computable in randomized polynomial-time, it remains to verify that the corresponding diversity embedding attains a distortion of $O(\log{n})$. That is, for any $A\in\mathcal{P}(X)$, we show that

\begin{equation*}
    \frac{1}{O(\log{n})}\delta_{\text{diam}}(A) \leq \delta_1(f(A)) \leq \delta_{\text{diam}}(A).
\end{equation*}

We choose $x,y\in A$ such that 
\begin{equation}\label{eqn:OurEmbedProof1}
    \delta_\text{diam}(A) = \max_{a,b\in A}d(a,b) = d(x,y).
\end{equation} We begin with the first inequality.

\begin{align*}
    \frac{1}{O(\log{n})}\delta_{\text{diam}}(A)  & = \frac{1}{O(\log{n})}\max_{u,v\in A} d(u,v) & \text{ by definition of a diameter diversity} \\
    & = \frac{1}{O(\log{n})} d(x,y) & \text{by choice of $x,y$} \\
    & \leq \|f(x) -f(y)\|_1 & \text{ by (\ref{eqn:LLR2Eqn1}), the metric embedding $f$} \\
    & = \sum_{i=1}^{O(\log^2{n})}|f_i(x) - f_i(y)| & \text{ by definition of the $\ell_1$ metric} \\
    & \leq\sum_{i=1}^{O(\log^2{n})}\max_{a,b\in A}|f_i(a)-f_i(b)| \\
    & = \delta_1(f(A)) & \text{by definition of the $\ell_1$ diversity}.
\end{align*}

This completes the first inequality. Then, for each $i\in \{1, 2, \ldots, O(\log^2{n})\}$ we let $a_i,b_i\in A$ be chosen such that
\begin{equation*}
    |d(a_i,A_i) - d(b_i,A_i)| = \max_{a,b\in A}|d(a,A_i) - d(b,A_i)|.
\end{equation*} We continue with the second inequality.

\begin{align*}
    \delta_1(f(A)) & = \sum_{i=1}^{O(\log^2{n})} \max_{a,b\in A}|f_i(a) - f_i(b)| & \text{by definition of the $\ell_1$ diversity} \\
    & =  \frac{1}{O(\log^2{n})}\sum_{i=1}^{O(\log^2{n})}\max_{a,b\in A}|d(a,A_i) - d(b,A_i)| & \text{by (\ref{eqn:LLR2Eqn2}), definition of $f$}\\
    & = \frac{1}{O(\log^2{n})}\sum_{i=1}^{O(\log^2{n})}|d(a_i,A_i) - d(b_i,A_i)| & \text{by choice of $a_i,b_i$'s} \\
    & \leq \frac{1}{O(\log^2{n})}\sum_{i=1}^{O(\log^2{n})}d(a_i,b_i) & \text{by Proposition \ref{prop:Frechet}} \\
    & \leq \frac{1}{O(\log^2{n})}\sum_{i=1}^{O(\log^2{n})}\max_{a,b\in A}d(a,b) & \text{by $a_i,b_i\in A$} \\
    & = \max_{a,b\in A}d(a,b) \\
    & = \delta_{\text{diam}}(A) & \text{by definition of the diameter diversity}.
\end{align*}

This completes the proof
\end{proof}

\iffalse
THis is just an observation.

\begin{theorem}
The diameter diversity embedding into the $\ell_1$ diversity from Theorem \ref{thm:DiamDivEmbeddingNew_orig} achieves an asymptotically optimal distortion.
\end{theorem}
%\begin{proof}
%This result follows immediately by Theorem %\ref{thm:DivLowerBound}, the $\ell_1$-diversity %embedding $\Omega(\log{n})$-distortion lower bound.
%\end{proof}
\fi

\subsubsection{A $k$-Diameter Diversity Embedding}

In this section we provide a proof of the fact that a $k$-diameter diversity can be embedded into the $\ell_1$ diversity with distortion $O(k\log{n})$.

At a high level, this $O(k\log{n})$-distortion embedding follows by approximating a $k$-diameter diversity with a diameter diversity, incurring an $O(k)$ approximation factor. Or in other words, embedding a $k$-diameter diversity into a diameter diversity with distortion $O(k)$. Then, the additional $O(\log{n})$ factor is incurred by embedding this diameter diversity into the $\ell_1$ diversity. 

\begin{theorem}\label{thm:kDiamToDiam}[see Appendix \ref{sec:omitted}]
Let $(X,\delta_{\text{$k$-diam}})$ be a $k$-diameter diversity, where $k\in \mathbb{Z}_{\geq 0}$. Then there is a diameter diversity into which $(X,\delta_{\text{$k$-diam}})$ can be embedded with distortion $O(k)$, in polynomial-time.
\end{theorem}

Then the main result of this section follows as a corollary of Theorems \ref{thm:kDiamToDiam} and \ref{thm:DiamDivEmbeddingNew_orig}. 

\begin{cor}\label{cor:kdiamdivembed}
Let $(X,\delta_{\text{$k$-diam}})$ be a $k$-diameter diversity, where $k\in \mathbb{Z}_{\geq 0}$. Then $(X,\delta_{\text{$k$-diam}})$ can be embedded into the $\ell_1$ diversity with distortion $O(k\log{n})$, in randomized polynomial-time.
\end{cor}

\begin{proof}
We let $f:(X,\delta_{\text{$k$-diam}})\to (X,\delta_\text{diam})$ be a polynomial-time $O(k)$-distortion embedding of $(X,\delta_\text{$k$-diam})$ into some diameter diversity $(X,\delta_\text{diam})$, due to Theorem \ref{thm:kDiamToDiam}. We let $g:(X,\delta_{diam})\to (\mathbb{R}^m,\delta_1)$ be a randomized polynomial-time $O(\log{n})$-distortion embedding of $(X,\delta_\text{diam})$ into the $\ell_1$ diversity for some dimension $m$, due to Theorem \ref{thm:DiamDivEmbeddingNew_orig}. Then the map $g \circ f$ is an embedding of $(X,\delta_{\text{$k$-diam}})$ into the $\ell_1$ diversity with distortion $O(k\log{n})$. Moreover, this embedding is computable in randomized polynomial-time. This completes the proof.
\end{proof}

\section{The Minimum Cost Hypergraph Steiner Problem}
\label{ch:HSP}

In this section we give an asymptotically optimal approximation algorithm for the minimum cost hypergraph Steiner problem. As far as we know,  this is  the first approximation algorithm for this problem (bar the case where $S=V$). 

\begin{defn}[Minimum Cost Hypergraph Steiner Problem (HSP)]\label{defn:HSP} 

Let $G = (V,E,w)$ be a hypergraph with nonnegative hyperedge weights $w:E\to\mathbb{R}_+$. For a set of \textit{Steiner nodes} $S\subseteq V$, we define $\mathcal{T}_S$ to be the collection of connected subhypergraphs of $G$ that contain the Steiner nodes $S$. Then the minimum cost hypergraph Steiner Problem, $HSP(G,S)$ is defined as

\begin{equation*}
HSP(G,S) = \min_{t\in \mathcal{T}_S} \sum_{U\in t} w(U).
\end{equation*} 

For convenience we may refer to the minimum cost hypergraph Steiner problem as simply the hypergraph Steiner problem, hence the use of the abbreviation HSP.

\end{defn}

This problem is a natural generalization of both the Steiner tree problem in graphs and the minimum cost spanning subhypergraph problem (MSSP) \cite{baudis2000approximating}, and yet surprisingly, it has not been explicitly investigated. We give the following logarithmic approximation. 

\begin{reptheorem}{thm:HSP_orig}
%\label{thm:HSP}
There exists a polynomial time $O(\log{n})$-approximation algorithm for the minimum cost hypergraph Steiner problem. Specifically, for a hypergraph $G = (V,E,w)$ with nonnegative hyperedge weights $w:E\to\mathbb{R}_+$ and for a set of Steiner nodes $S\subseteq V$, $HSP(G,S)$ can be approximated up to a factor of $O(\log{|S|})$ in polynomial time.
\end{reptheorem}

We remark that this is tight due to the reduction from set cover used in \cite{baudis2000approximating}.

Before proving the result, we note that the minimum cost hypergraph Steiner problem is a computational bottleneck for our approach of utilizing diversity embeddings for the sparsest cut problem in hypergraphs (see Appendix~\ref{sec:applications} for more details).
It arises in the following two computational problems.

\begin{enumerate}
    \item Given a hypergraph Steiner diversity $(V,\delta_\mathcal{H})$ defined by a hypergraph $H=(V,E,w)$ with nonnegative hyperedge weights $w:E\to\mathbb{R}_+$, for any $S\subseteq V$, the diversity value $\delta_\mathcal{H}(S)$ is precisely defined as $\delta(S) = HSP(H,S)$. If one is given $(V,\delta_\mathcal{H})$ by the hypergraph $H$ then to (approximately) query $\delta_\mathcal{H}(S)$ one must (approximately) solve $HSP(H,S)$.
    %For the application to sparsest cut in hypergraphs, the hypergraph %Steiner diversity that emerges from the sparsest cut LP Relaxation %(\ref{eqn:LP1}) is defined implicitly by a hyperedge-weighted %hypergraph. Thus, to (approximately) solve the sparsest cut problem %in hypergraphs via diversity embeddings it is necessary to %(approximately) solve the minimum cost hypergraph Steiner problem.
    
    \item The sparsest cut LP Relaxation (\ref{eqn:LP1}) has polynomially many variables but exponentially many constraints. Specifically, for each demand $S\in E_H$, there may be an exponential number of constraints of the form
    
    \begin{equation*}
        \sum_{U\in t} d_U \geq y_S, \forall t\in \mathcal{T}_{(G,S)}.
    \end{equation*}
    
   Hence, obtaining polytime solutions to the relaxation requires,   a polynomial time (approximate) separation oracle for these constraints. This amounts to solving a minimum cost hypergraph Steiner problem for each $S$. 
    
\end{enumerate}

\iffalse
\begin{theorem}\label{thm:HSPLowerBound}
The minimum cost hypergraph Steiner problem cannot be approximated to a factor smaller than $\Omega(\log{n})$ in polynomial time unless P=NP. Specifically, for a hypergraph $G = (V,E,w)$ with nonnegative hyperedge weights $w:E\to\mathbb{R}_+$ and for a set of Steiner nodes $S\subseteq V$, $HSP(G,S)$ cannot be approximated up to a factor smaller than $\Omega(\log{|S|})$ in polynomial time unless P=NP.
\end{theorem}

\subsection{An Optimal Algorithm for the Minimum Cost Hypergraph Steiner Problem}\label{sec:HSP2}

In this section we provide proofs of Theorems \ref{thm:HSP_orig}
and \ref{thm:HSPLowerBound}.
\fi

\subsection{Proof of Theorem \ref{thm:HSP_orig}}

We proceed by reducing to the minimum cost node-weighted Steiner tree problem defined as follows.

\begin{defn}
Let $G=(V,E,w)$ be a graph with nonnegative node weights $w:V\to\mathbb{R}_+$. For a set of Steiner nodes $S\subseteq V$ we define $\mathcal{T}_S$ to be the set of minimally connected subgraphs of $G$ containing the Steiner nodes $S$, or simply subtrees of $G$ spanning $S$.  Then the minimum cost node-weighted Steiner tree problem (NSTP), NSTP$(G,S)$, is defined as

\begin{equation*}
    \text{NSTP}(G,S) = \min_{t\in\mathcal{T}_S} \sum_{v \in \cup_{e\in t} e} w(v).
\end{equation*}
\end{defn}

Klein and Ravi \cite{kleinravi} establish the following approximation algorithm for this problem. 

\begin{theorem}[Restatement of Theorem 1.1 of \cite{kleinravi}]\label{thm:KleinRavi} Let $G=(V,E,w)$ be a graph with nonnegative node weights $w:V\to\mathbb{R}_+$. For a set of Steiner terminals $S\subseteq V$ there is a polynomial time $O(\log{|S|})$-approximation algorithm for the minimum-cost node-weighted Steiner tree problem NSTP$(G,S)$.
\end{theorem}

We now reduce the minimum-cost hypergraph Steiner problem to the minimum-cost node-weighted Steiner tree problem.

\begin{lemma}\label{lemma:HSPReduction}
Let $G=(V,E,w)$ be a hypergraph with nonnegative hyperedge weighted $w:E\to \mathbb{R}_+$ and let $S\subseteq V$ be a set of Steiner nodes. Then the minimum-cost hypergraph Steiner problem HSP$(G,S)$ can be reduced to an instance of the minimum-cost node-weighted Steiner tree problem in polynomial time.
\end{lemma}

\begin{proof}
We first construct the instance of the minimum-cost node-weighted Steiner tree problem to `simulate' HSP$(G,S)$. We may trivially assume $|S| \geq 2$. For each hyperedge $U\in E$ we define an associated node $v_U$ and we denote these nodes by  $V_E := \{v_U: U\in E\}$

We then define $V' = V\cup V_E$. Next, for each $v_U\in V_E$ we create an edge between $v_U$ and each $v\in U$. Specifically, we define this collection of edges $E'$ as

\begin{equation*}
    E' = \{\{v,v_U\}: v\in U\}.
\end{equation*}

Finally, we define the nonnegative node-weights $w':V'\to\mathbb{R}_+$ as 

\begin{equation*}
    w'(v) = \begin{cases}
        0 & \text{ if } v\in V \\
        w(U) & \text{ if } v=v_U \text{ for some } v_U\in V_E.
    \end{cases}
\end{equation*}

Thus, our node-weighted graph is defined to be $G'=(V',E',w')$ and the instance of the minimum-cost node-weighted Steiner tree problem that we are reducing to is NSTP$(G',S)$. It is easy to see that this reduction can be computed in polynomial time with respect to $G$. It remains to prove the correctness of this reduction. 

We first prove that HSP$(G,S) \geq $ NSTP$(G',S)$. Let $t$ be a subset of the hyperedges of $G$ that correspond to a connected subhypergraph of $G$ that contains the Steiner nodes $S$. We define a feasible solution to NSTP$(G',S)$ as

\begin{equation*}
    t' = \{\{v,v_U\}: v\in U, U \in t\}.
\end{equation*}

Then the objective value of this solution to NSTP$(G',S)$ is 

\begin{equation*}
    \sum_{v\in \cup_{e\in t'}e} w'(v) = \sum_{U\in t} w'(v_U) = \sum_{U\in t} w(U).
\end{equation*} where the first equality follows by the fact that $w(v) = 0$ for all $v\in V$ and the second equality follows by the fact that $w'(v_U)=  w(U)$ for all $v_U\in V_E$. Thus, the objective value of $t'$ is at most that of $t$ in HSP$(G,S)$. It remains to show that $t'$ is a feasible solution to NSTP$(G',S)$, that is $t'$ is connected and $t'$ contains the nodes $S$. 

We first prove the latter condition on $t'$. Let $v\in S$ be arbitrary. Since $t$ is subhypergraph of $G$ that contains the nodes $S$ there is some $U\in t$ such that $v\in U$. Since $t'$ is connected and $|S| \geq 2$, there is an edge $(v,v_U)\in t'$ and so $t'$ contains $v$. 

As for the former condition on $t'$, we let $v,v' \in \cup_{e\in t'} e$ be two arbitrary but distinct nodes of the graph defined by $t'$. Without loss of generality, we can assume that $v,v' \in V$ since 
there is always an edge from any $v_U$ to a $v \in U$.  Since $t$ is a connected subhypergraph of $G$, there is a sequence of hyperedges $U_1,U_2,\ldots,U_k$ such that $v\in U_1$, $v'\in U_k$, and for each $i \in \{1,2,\ldots,k-1\}$, $U_i \cap U_{i+1} \neq \emptyset$. By the last property of $U_1,U_2,\ldots,U_k$, we can create a sequence of nodes $v_1,v_2,\ldots,v_k$ such that for each $i\in \{1,2,\ldots,k-1\}, v_i \in U_i\cap U_{i+1}$. Then, by the construction of $t'$, it follows that the sequence 

\begin{equation*}
    v,v_{U_1},v_1,v_{U_2},v_2,\ldots,v_{k-1},v_{U_k},v'
\end{equation*} of nodes of $G'$ forms a path in the subgraph $t'$ of $G'$. Hence HSP$(G,S) \geq $ NSTP$(G',s)$.

We now prove the other inequality, HSP$(G,S) \leq $ NSTP$(G',s)$. Let $t'$ be a connected subgraph of $G'$ containing the Steiner nodes $S$.
 We define a feasible solution to HSP$(G,S)$ as  $t = \{U: v_U \in V_E\cap t'\}$. Then the objective value of this solution to HSP$(G,S)$ is

\begin{equation*}
    \sum_{U\in t}w(U) = \sum_{v_U\in V_E\cap t'}w'(v_U) = \sum_{v\in t'} w'(v)
\end{equation*} where the first equality follows by the fact that $w(U) = w'(v_U)$ for all $v_U\in V_E\cap t'$ and the second equality follows by the fact that $w'(v) = 0$ for all $v\in V$. Thus, the objective value of $t$ to HSP$(G,S)$ is at least as small as that of $t'$ to NSTP$(G',S)$. It remains to show that $t$ is a feasible solution to HSP$(G,S)$, that is, $t$ is connected and $t$ contains the Steiner nodes $S$.

We first prove the latter condition on $t$. Let $v\in S$ be arbitrary. Due to the fact that $t'$ contains the Steiner nodes $S$, and in particular $v$, and that $t'$ is connected there must be some edge $(v,v')\in t'$. Furthermore, $G'$ is a bipartite graph by construction, with the partitions $V\subseteq V'$ and $V_E\subseteq V'$, and since $v\in S\subseteq V$ then $v'\in V_E$. Therefore, $v' = v_U$ for some $v_U\in V_E$ where $U$ is a hyperedge of $G$ that contains the node $v$. Finally, $U\in t$ by construction of $t$ and so the subhypergraph $t$ contains the node $v$. 

As for the former condition on $t$, we let $v,v'\in \cup_{U\in t}U$ be two arbitrary but distinct nodes of the subhypergraph defined by $t$. Since $t'$ is a connected bipartite subgraph of $G'$ there is a sequence of nodes $v_1,v_{U_1}, v_2,v_{U_2},\ldots, v_{U_{k-1}},v_k$ where $v= v_1$, $v'=v_k$ and for each $i\in\{1,2,\ldots,k-1\}$ there are edges $(v_i,v_{U_i})$ and $(v_i,v_{U_{i+1}})\in t'$. Hence $v_i \in U_i\cap U_{i+1}$. Or in other words, the sequence of hyperedge $U_1,U_2,\ldots, U_{k-1}$ forms a path from $v\in U_1$ to $v\in U_{k-1}$, where the membership of $v\in U_1$ and $v\in U_{k-1}$ follows by construction of the edges of $G'$. Hence $t$ is connected
and we have proved HSP$(G,S) \leq$ NSTP$(G',S)$.

\end{proof}

Theorem \ref{thm:HSP_orig} now follows immediately by this reduction and the Klein and Ravi algorithm, Theorem \ref{thm:KleinRavi}.

\begin{singlespace}
\raggedright
\bibliographystyle{abbrvnat}
\bibliography{biblio}
\end{singlespace}

\appendix

\section{Omitted Proofs}
\label{sec:omitted}

Bryant and Tupper~\cite[Proposition~2.1(2)]{Bryant_2012} prove the following result for diversities. Their argument applies unchanged to pseudodiversities, and we include the proof for completeness.

\begin{repproposition}{prop:Monotonicity}
(Pseudo) diversities are monotone increasing. That is, if $(X, \delta)$ is a (pseudo) diversity then for any $A,B\in \mathcal{P}(X)$ we have that

\begin{equation*}
    A\subseteq B \Rightarrow \delta(A) \leq \delta(B).
\end{equation*}
\end{repproposition}

\begin{proof}
We let $A,B\in \mathcal{P}(X)$ be arbitrary and we assume that $A\subseteq B$. Then we let $B\setminus A = \{x_1, x_2, \ldots, x_k\}$. We define $A_0 = A$ and for each $i\in \{1, 2, \ldots, k\}$ we define $A_i = A \cup \{x_1, x_2, \ldots, x_i\}$. For any arbitrary $i\in \{0, 1, 2, \ldots, k-1\}$ we show that $\delta(A_i) \leq \delta(A_{i+1})$. 

\begin{align*}
    \delta(A_i) & \leq \delta(A_i \cup \{x_{i+1}\}) + \delta(\{x_{i+1}\}) & \text{by the triangle inequality of diversities} \\
& = \delta(A_i \cup \{x_{i+1}\}) & \text{$\delta(\{x_{i+1}\}) = 0$ since $|\{x_{i+1}\}| = 1$} \\
& = \delta(A_{i+1}).
\end{align*}

Since $A_k = B$, the above result implies that 

\begin{equation*}
    \delta(A) \leq \delta(A_1) \leq \delta(A_2) \leq \ldots \leq \delta(A_{k-1}) \leq \delta(B)
\end{equation*} thus completing the proof.
\end{proof}

\begin{reptheorem}{thm:MinimalityOfkDiameterDiversity}
Given a (pseudo) diversity $(X,\delta)$ and $k\in \mathbb{Z}_{\geq 0}$ we let $(X,\delta_{\text{$k$-diam}})$ be the $k$-diameter (pseudo) diversity of $(X,\delta)$. Then for any $(X,\delta')\in \mathcal{D}_{(X,\delta,k)}$ and any $A\in\mathcal{P}(X)$ it follows that

\begin{equation*}
    \delta_{\text{$k$-diam}}(A) \leq \delta'(A).
\end{equation*} Or in other words $(X,\delta_{\text{$k$-diam}})$ is the minimal (pseudo) diversity of the family $\mathcal{D}_{(X,\delta,k)}$.
\end{reptheorem}

\begin{proof}
For an arbitrary $k\in\mathbb{Z}_{\geq 0}$, we let $(X,\delta_{\text{$k$-diam}})$ be the $k$-diameter diversity of $(X,\delta)$. First, we argue that $(X,\delta_{\text{$k$-diam}})$ is in fact a member of $\mathcal{D}_{(X,\delta,k)}$. We consider an arbitrary $A\in\mathcal{P}(X)$ such that $|A| \leq k$. Then it follows that

\begin{align*}
\delta_{\text{$k$-diam}}(A) & = \max_{B\subseteq A: |B|\leq k}\delta(B) & \text{by definition of a $k$-diameter diversity} \\
& = \delta(A) & \text{by $|A|\leq k$ and diversities being increasing, Proposition \ref{prop:Monotonicity}}.
\end{align*}

Hence, $\forall A\in\mathcal{P}(X)$ such that $|A| \leq k$ it follows that $\delta_{\text{$k$-diam}}(A) = \delta(A)$ and so $(X,\delta_{\text{$k$-diam}})\in \mathcal{D}_{(X,\delta,k)}$. Next, we establish the minimality of $(X,\delta_{\text{$k$-diam}})$ with regards to this family of diversities. We let $(X,\delta') \in \mathcal{D}_{(X,\delta,k)}$ and $A\in\mathcal{P}(X)$ be arbitrary. Then it follows that

\begin{align*}
\delta_{\text{$k$-diam}}(A) & = \delta(B) & \text{for some $B\subseteq A$ such that $|B|\leq k$} \\
& = \delta'(B) & \text{by definition of $(X,\delta')\in \mathcal{D}_{(X,\delta,k)}$} \\
& \leq \delta'(A) & \text{by diversities being increasing, Proposition \ref{prop:Monotonicity}}.
\end{align*}

This completes the proof.

\end{proof}

\begin{replemma}{lemma:subadditive}
Let $X$ be a set and let $f:2^X\to\mathbb{R}$ be a nonnegative, increasing, and subadditive set function. Then $(X,\delta)$ is a pseudo-diversity where $\delta$ is defined as 

\begin{equation*}
    \delta(A) = \begin{cases}  f(A) & \text{if $|A| \geq 2$} \\ 0 & \text{otherwise}. \end{cases}
\end{equation*}
\end{replemma}

\begin{proof}
Let $\delta$ be defined as in the lemma statement. We begin by showing that $(X,\delta)$ satisfies the first two axioms of a pseudo-diversity. By the nonnegativity of $f$, $\delta$ is likewise nonnegative by its construction, thus $(X,\delta)$ satisfies the first axiom of pseudo-diversities. Likewise by construction, $\delta(A) = 0$ if $|A| \leq 1$ and thus $(X,\delta)$ satisfies the second axiom. 

In order to prove that $(X,\delta)$ satisfies the third axiom of a pseudo-diversity we first establish that $\delta$ is an increasing set function. Let $A\subseteq B\in\mathcal{P}(X)$ be arbitrary. For the case where $|A| \leq 1$ then 

\begin{equation*}
    \delta(A) = 0 \leq \delta(B)
\end{equation*} by the nonnegativity of $\delta$. As for the case where $|A| \geq 2$ then 

\begin{equation*}
    \delta(A) = f(A) \leq f(A\cup B) = \delta(A\cup B)
\end{equation*} where the two equalities follow by the fact that $|A|,|A\cup B| \geq 2$ and the inequality follows by the fact that $f$ is increasing. This concludes the proof that $\delta$ is an increasing set function.

We now prove that $(X,\delta)$ satisfies the third axiom of pseudo-diversities. Let $A,B,C\subseteq X$ be arbitrary where $C\neq \emptyset$. We consider several cases.

\begin{enumerate}
    \item The first case is where one of $A$ and $B$ is empty. Without loss of generality, we assume that $A = \emptyset$. Then,
    \begin{align*}
        \delta(A\cup B) & = \delta(B) \\
        & \leq \delta(B\cup C) & \text{by $\delta$ being increasing} \\
        & \leq \delta(A\cup C) + \delta(B\cup C) & \text{by $\delta$ being nonnegative}.
    \end{align*}
   
    \item The second case that we consider is where $|A| = |B| = 1$. There are two subcases that we consider. The first is where $A=B$ in which case $|A\cup B| = 1$ and so
    
    \begin{align*}
         \delta(A\cup B) & = 0 & \text{by $|A\cup B| = 1$} \\
         & \leq \delta(A\cup C) + \delta(B\cup C) & \text{by $\delta$ being nonnegative}.
    \end{align*}
    
    Then the second subcase is where $A\neq B$ and so $|A\cup B| \geq 2$. Additionally, we divide this subcase into two more sub-subcases. The first sub-subcase is where $A\cap C = \emptyset$ and $B\cap C = \emptyset$, in which case $|A\cup C|,|B\cup C|\geq 2$ and so 
    
    \begin{align*}
         \delta(A\cup B)& = f(A\cup B) & \text{by $|A\cup B| \geq 2$}  \\
         & \leq f(A) + f(B) & \text{by subadditivity of $f$ } \\
         & \leq f(A\cup C) + f(B\cup C) & \text{by $f$ being increasing} \\
         & = \delta(A\cup C) + \delta(B\cup C) & \text{ by $|A\cup C|,|B\cup C| \geq 2$}.
    \end{align*}
    
    The second sub-subcase is where, without loss of generality, $B \subseteq C$, and so 
    
    \begin{align*}
         \delta(A\cup B)& \leq \delta(A\cup C) & \text{by $\delta$ being increasing} \\
         & \leq \delta(A\cup C) + \delta(B\cup C) & \text{by $\delta$ being nonnegative}.
    \end{align*}
    
    \item The final case is where, without loss of generality, $|A| \geq 2$ and $|B| \geq 1$. Furthermore, we consider two subcases. The first is where $|B\cup C| = 1$ in which case $B=C$ and so
    
    \begin{align*}
        \delta(A\cup B) & =  f(A\cup B) & \text{by $|A\cup B|\geq 2$} \\
        & = f(A\cup C) & \text{by $B=C$} \\
        & = \delta(A\cup C) & \text{by $|A\cup C| \geq 2$} \\
        & \leq \delta(A\cup C) + \delta(B\cup C) & \text{by $\delta$ being nonnegative}.
    \end{align*}
    
    Then the second subcase is where $|B\cup C| \geq 2$. Then,
    
    \begin{align*}
         \delta(A\cup B) & = f(A\cup B) & \text{by $|A\cup B|\geq 2$} \\
         & \leq f(A) + f(B) & \text{by $f$ being subadditive} \\
         & \leq f(A\cup C) + f(B\cup C) & \text{by $f$ being increasing} \\
        & = \delta(A\cup C) + \delta(B\cup C) & \text{by $|A\cup C|,|B\cup C| \geq 2$}.
    \end{align*}
    
\end{enumerate}

This completes the proof.
\end{proof}

\begin{reptheorem}{thm:kDiamToDiam}
Let $(X,\delta_{\text{$k$-diam}})$ be a $k$-diameter diversity, where $k\in \mathbb{Z}_{\geq 0}$. Then there is a diameter diversity into which $(X,\delta_{\text{$k$-diam}})$ can be embedded with distortion $O(k)$, in polynomial-time.
\end{reptheorem}

\begin{proof}
We let $(X,d)$ be the induced metric space of $(X,\delta_{\text{$k$-diam}})$. We let $(X,\delta_{\text{diam}})\in\mathcal{D}_{(X,d)}$ be the corresponding diameter diversity whose induced metric space is $(X,d)$. It suffices to show that $(X,\delta_{\text{$k$-diam}})$ embeds into $(X,\delta_{\text{diam}})$ with distortion $k$. That is, for any $A\in\mathcal{P}(X)$ it suffices to show that

\begin{equation*}
    \delta_{\text{diam}}(A) \leq \delta_{\text{$k$-diam}} \leq k \delta_{\text{diam}}.
\end{equation*}

We let $A\in\mathcal{P}(X)$ be arbitrary and we choose $B\subseteq A$ where $|B|\leq k$ and \begin{equation*}
    \delta_{\text{$k$-diam}}(A) = \delta_{\text{$k$-diam}}(B).
\end{equation*} Without loss of generality we enumerate the elements of $B$ as \begin{equation*}
    B = \{v_1,v_2,\ldots,v_j\}
\end{equation*} where $j \leq k$. Then it follows that

\begin{align*}
    \delta_{\text{diam}}(A) & \leq \delta_{\text{$k$-diam}}(A) & \text{by minimality of the diameter diversity, Theorem \ref{thm:ExtremalResults}} \\
    & = \delta_{\text{$k$-diam}}(B) & \text{by choice of $B$} \\
    & \leq \sum_{i=2}^j d(v_1,v_i) & \text{by the triangle inequality}\\
    & \leq \sum_{i=2}^j\max_{u,v\in B} d(u,v) \\
    & \leq k \max_{u,v\in B}d(u,v) & \text{by $j\leq k$} \\
    & = k \delta_{\text{diam}}(B) & \text{by definition of the diameter diversity} \\
    & \leq k \delta_{\text{diam}}(A) & \text{by $B\subseteq A$ and diversities being increasing, Proposition \ref{prop:Monotonicity}}.
\end{align*}

The induced metric space of $(X,d)$ of $(X,\delta_{\text{$k$-diam}})$ can be computed in polynomial-time. Then the corresponding diameter diversity $(X,\delta_\text{diam})$ can be computed in polynomial-time given the $(X,d)$. Hence, this embedding is polynomial-time computable. This completes the proof.
\end{proof}

\begin{replemma}{lemma:ISLemma}
Let $G = (V,E)$ be a graph. We define the independent set function $f_{IS}:2^V\to \mathbb{Z}_{\geq 0}$ as

\begin{equation*}
    f_{IS}(A) = \max\{|S|: \text{$S\subseteq A$, $S$ is an independent set of $G$}\}.
\end{equation*}

Then the set function $f_{IS}$ is nonnegative, increasing, and subadditive.
\end{replemma}

\begin{proof}
By construction, $f_{IS}$ is clearly nonnegative. As for the increasing property of $f_{IS}$, it suffices to show that for any any arbitrary $A\subseteq B\subseteq V$, $f_{IS}(A) \leq f_{IS}(B)$. We let $A\subseteq B\subseteq V$ be arbitrary and we suppose that $S\subseteq A$ is an independent set of $G$ such that $f_{IS}(A) = |S|$. Since $S\subseteq A\subseteq B$, it also follows that $|S| \leq f_{IS}(B)$. Hence $f_{IS}(A) \leq f_{IS}(B)$ and thus $f_{IS}$ is increasing.

Next, to prove that $f_{IS}$ is subadditive we first argue that independent sets of $G$ are downwards closed. That is, if $S\subseteq V$ is an independent of $G$, then for any $S'\subseteq S$ it follows that $S'$ is an independent set of $G$. Let $u,v\in S'$ be arbitrary, since $u,v\in S'\subseteq S$ and $S$ is an independent set of $G$ then $(u,v)\not\in E$ and so $S'$ is also an independent set of $G$. 

Now we proceed to prove that $f_{IS}$ is subadditive. Let $A,B\subseteq V$ be arbitrary. It suffices to show that $f_{IS}(A\cup B) \leq f_{IS}(A) + f_{IS}(B)$. We define $A' = A \setminus B$ and $B' = B$. Then $A'\cap B' = \emptyset$ and $A'\cup B' = A\cup B$. Hence,

\begin{equation}\label{eqn:ISProof1}
    f_{IS}(A\cup B) = f_{IS}(A'\cup B').
\end{equation}

We let $S_{A'\cup B'}\subseteq A'\cup B', S_{A'}\subseteq A', S_{B'}\subseteq B'$ be independent sets of $G$ such that

\begin{equation}\label{eqn:ISProof2}
    f_{IS}(A'\cup B') = |S_{A'\cup B'}|, \text{ } f_{IS}(A') = |S_{A'}|, \text{ } f_{IS}(B') = |S_{B'}|.
\end{equation}

Since $S_{A'\cup B'} \subseteq A'\cup B'$ then

\begin{equation}\label{eqn:ISProof3}
    S_{A'\cup B'} = (S_{A'\cup B'} \cap A') \cup (S_{A'\cup B'} \cap B').
\end{equation}

Moreover, since $A'\cap B' = \emptyset$ then 

\begin{equation}\label{eqn:ISProof4}
    |S_{A'\cup B'}| = |S_{A'\cup B'} \cap A'| + |S_{A'\cup B'} \cap B'|.
\end{equation}

We also note that by the downwards closed property of independent sets, $S_{A'\cup B'}\cap A'$ and $S_{A'\cup B'}\cap B'$ are independent sets of $G$. Furthermore, this implies that

\begin{equation}\label{eqn:ISProof5}
    |S_{A'\cup B'}\cap A'| \leq f_{IS}(A') \text{ and } |S_{A'\cup B'}\cap B'| \leq f_{IS}(B').
\end{equation}

Putting everything together, we have that

\begin{align*}
    f_{IS}(A\cup B) & = f_{IS}(A'\cup B') & \text{ by (\ref{eqn:ISProof1}}) \\
    & = |S_{A'\cup B}| & \text{ by (\ref{eqn:ISProof2})} \\
    & = |S_{A'\cup B'} \cap A'| + |S_{A'\cup B'} \cap B'| & \text{ by (\ref{eqn:ISProof4})} \\
    & \leq f_{IS}(A') + f_{IS}(B') & \text{ by (\ref{eqn:ISProof5})} \\
    & \leq f_{IS}(A) + f_{IS}(B) & \text{ by the fact that $f_{IS}$ is increasing}.
\end{align*}
This completes the proof.
\end{proof}

\begin{reptheorem}{thm:InapproxResult_orig}
%\label{thm:InapproxResultAppendix}
For any $p \geq 0$ and for any $\epsilon > 0$, there does not exist a polynomial-time diversity $\ell_1$ embedding that has oracle access to a diversity and queries it only on sets of cardinality at most $O(\log{n^p})$ with a distortion of $O(n^{1-\epsilon})$, unless P=NP.
\end{reptheorem}

\begin{proof}
For this proof, we assume that P$\neq$NP, and for the sake of contradiction, we assume the negation of the theorem statement. That is, we suppose that for some $p \geq 0$ and for some $\epsilon > 0$ there is a polynomial-time diversity embedding into $\ell_1$ that queries the diversity on sets of cardinality at most $O(\log{n^p})$ with a distortion of $O(n^{1-\epsilon})$.

We let $G = (V,E)$ be an arbitrary graph with $|V| = n$. We let $(V,\delta_{IS})$ be the corresponding independent set diversity as defined by Definition \ref{defn:IndependentSetDiversity}. $(V,\delta_{IS})$ is defined implicitly by the graph $G$ and, at this point, we are not insisting on being able to explicitly compute $\delta_{IS}(A)$ for any $A\subseteq V$.  We note that by definition of $(V,\delta_{IS})$ we have that

\begin{equation*}
    \delta_{IS}(V) = \max\{|S|: \text{$S\subseteq V$, $S$ is an independent set of $G$}\} = \text{ISP}(G).
\end{equation*}

Therefore, any approximation of $\delta_{IS}(V)$ yields an approximation of the instance of the independent set problem, ISP$(G)$. In particular, we show how an algorithm that embeds a diversity into $\ell_1$ yields an approximation algorithm for $\delta_{IS}(V)$ and in turn for ISP$(G)$.

We let $f:V\to \mathbb{R}^d$, for some $d\in\mathbb{Z}_{\geq 0}$, be the diversity embedding of $(V,\delta_{IS})$ into $(\mathbb{R}^d, \delta_1)$ given by our algorithm at the beginning of this proof. By assumption we can compute $f$ in polynomial-time, provided we can query $\delta_{IS}(A)$ for sets $A\subseteq V$ where $|A| \in O(\log{n^p})$ in polynomial-time. This requirement on the polynomial time computability of $\delta_{IS}(A)$ is necessary since $f$ requires (polynomia time) oracle access to $\delta_{\text{IS}}(A)$. Although, $(V,\delta_{IS})$ is defined implicitly by the graph $G$, we can compute $\delta_{IS}(A)$ for every $A\subseteq V$ where $|A|\in O(\log{n^p})$ in polynomial-time, specifically $O(n^p)$, using a brute-force enumeration of all independent sets of $G$ that are subsets of $A$. Therefore, $f$ is computable in polynomial-time.

Furthermore, the distortion of the embedding $f$ is of the factor $O(n^{1-\epsilon})$. Specifically, we have the following guarantee on the value of $\delta_{IS}(V)$ with respect to $\delta_1(f(V))$,

\begin{equation*}
    \frac{1}{c_1}\delta_{IS}(V) \leq \delta_1(f(V)) \leq c_2\delta_{IS}(V)
\end{equation*} where $c_1,c_2 > 0 $ and $c_1c_2 = O(n^{1-\epsilon})$. Given that $\delta_{IS}(V) = \text{ISP}(G)$ and that $f$ is computable in polynomial-time, this implies that we have a polynomial-time approximation algorithm for the independent set problem with an approximation factor of $O(n^{1-\epsilon})$. However, this contradicts Theorem \ref{thm:InapproxIS}. This completes the proof.

\end{proof}

\section{Applications of the Hypergraph Steiner Problem}
\label{sec:applications}

In this section we provide two applications of our approximation algorithm for the minimum cost hypergraph Steiner problem, Theorem \ref{thm:HSP_orig}. The first application is that of approximately computing a hypergraph Steiner diversity given a hyperedge-weighted hypergraph.

\begin{cor}\label{cor:HSteinerDivApproxAlg}
Let $H=(V,E,w)$ be a hypergraph with node set $V$, hyperedge set $E$, and nonnegative hyperedge weights $w:E\to\mathbb{R}_+$. Let $(V,\delta_\mathcal{H})$ be the corresponding hypergraph Steiner diversity defined by $H$. Then for any $A\in\mathcal{P}(V)$, $\delta(A)$ can be computed up to a factor of $O(\log{|A|})$ in polynomial time. 
\end{cor}
\begin{proof}
This corollary follows immediately from Theorem \ref{thm:HSP_orig} and by the fact that for any $A\in\mathcal{P}(V)$ we have that \begin{equation*}
    \delta(A) = \min_{t\in \mathcal{T}_{(H,A)}}\sum_{U\in t}w(U) = HSP(H,A)
\end{equation*} and so computing $\delta(A)$ is equivalent to solving $HSP(H,A)$.
\end{proof}

The second application of Theorem \ref{thm:HSP_orig} is that LP Relaxation (\ref{eqn:LP1}) can be approximated up to an approximation factor of $O(\log{r_H})$ in polynomial time.

\begin{cor}\label{cor:LPRelaxApproxAlg}
Let $G=(V,E_G,w_G)$ be a supply hypergraph and let $H=(V,E_H,w_H)$ be a demand hypergraph with rank $r_H$. Then the sparsest cut LP Relaxation (\ref{eqn:LP1}) can be be approximated up to a factor of $O(\log{r_H})$ in polynomial time. 
\end{cor}

\begin{proof}
LP Relaxation (\ref{eqn:LP1}) has polynomially many variables and, in general, exponentially many constraints. Specifically, for each $S\in E_H$ the following set of constraints may be exponentially large

\begin{equation}\label{eqn:HSPEqn1}
    \{\sum_{U\in t} d_U \geq y_S\}_{t\in \mathcal{T}_{(G,S)}}.
\end{equation} Separating over these constraints amounts to the decision problem \begin{equation*}
    HSP(G,S) = \min_{t\in\mathcal{T}_{(G,S)}}\sum_{U\in t}d_U \geq y_S.
\end{equation*} We can use Theorem \ref{thm:HSP_orig} to approximate HSP(G,S) up to an $O(\log{r_H})$ factor, as $|S|\leq r_H$. Hence, up to a factor of $O(\log{r_H})$, we can verify whether the set of constraints (\ref{eqn:HSPEqn1}) are approximately satisfied, and if not, we can find an approximate separating hyperplane. Then by the ellipsoid algorithm we can solve the LP Relaxation (\ref{eqn:LP1}) up to an $O(\log{r_H})$ approximation factor in polynomial time -- see related work \cite{jain2003packing}.
\end{proof}

%    7. Index
% See the makeindex package: the following page provides a quick overview
% <http://www.image.ufl.edu/help/latex/latex_indexes.shtml>

\end{document}